\documentclass[pra,onecolumn,groupedaddress,amsmath,amssymb,a4paper]{revtex4}
\usepackage{geometry}
\usepackage{graphicx}
\usepackage{mathrsfs}
\usepackage{amssymb}
\usepackage{amsmath}
\usepackage{amsthm}
\usepackage{amsbsy}
\usepackage{bm}
\usepackage{hyperref}

\newtheorem{proposition}{Proposition}
\newtheorem{corollary}{Corollary}

\theoremstyle{definition}

\newtheorem{examplex}{Example}
\newenvironment{example}
  {\pushQED{\qed}\examplex}
  {\popQED\endexamplex}

\newcommand{\ave}[1]{\langle #1 \rangle}
\newcommand{\bra}[1]{\langle #1|}
\newcommand{\ket}[1]{| #1 \rangle }

\begin{document}

\title{Capacity of trace decreasing quantum operations and
superadditivity of coherent information for a generalized erasure
channel}

\author{Sergey N. Filippov}

\affiliation{Steklov Mathematical Institute of Russian Academy of
Sciences, Gubkina St. 8, Moscow 119991, Russia}

\begin{abstract}
Losses in quantum communication lines severely affect the rates of
reliable information transmission and are usually considered to be
state-independent. However, the loss probability does depend on
the system state in general, with the polarization dependent
losses being a prominent example. Here we analyze biased trace
decreasing quantum operations that assign different loss
probabilities to states and introduce the concept of a generalized
erasure channel. We find lower and upper bounds for the classical
and quantum capacities of the generalized erasure channel as well
as characterize its degradability and antidegradability. We reveal
superadditivity of coherent information in the case of the
polarization dependent losses, with the difference between the
two-letter quantum capacity and the single-letter quantum capacity
exceeding $7.197 \cdot 10^{-3}$ bits per qubit sent, the greatest
value among qubit-input channels reported so far.
\end{abstract}

\maketitle


\section{Introduction} \label{section-introduction}

Transmission of information through noisy quantum communication
lines has fascinating properties. Measurements in the basis of
entangled states enable extracting more classical information as
compared to the individual measurements of each information
carrier~\cite{holevo-1998,schumacher-1997}. Moreover, encoding
classical information into entangled states would give even better
communication rates~\cite{hastings-2009}. In addition to sending
classical information, it is possible to reliably transmit quantum
information and create entanglement between the sender and the
receiver even if the communication line is noisy, thus opening an
avenue for quantum
networking~\cite{lloyd-1997,barnum-1998,devetak-2005}. A typical
kind of noise in quantum communication lines is the loss of
information carriers, e.g., photons. For continuous-variable
quantum states this effect is intrinsically included in the
description~(see, e.g.,~\cite{weedbrook-2012}), whereas for
discrete-variable quantum states this effect is usually described
by an erasure channel~\cite{grassl-1997,bennett-1997}. For
instance, if the information is encoded into polarization degrees
of freedom of single photons (that can be potentially entangled
among themselves), then the erasure channel accounts for the loss
of photons in the line, with the probability to lose a
horizontally polarized photon being the same as the probability to
lose a vertically polarized photon. Additional effects of
decoherence are taken into account by concatenating the erasure
channel with the decoherence map, e.g., a combination of the
dephasure and the loss results in the so-called ``dephrasure''
channel~\cite{leditzky-2018} and a combination of the amplitude
damping channel and the loss is considered in
Ref.~\cite{siddhu-2020}, with the phenomenon of superadditivity of
coherent information being observed in the both
cases~\cite{leditzky-2018,siddhu-2020}. In particular, the
two-letter quantum capacity exceeds the single-letter quantum
capacity by about $2.5 \cdot 10^{-3}$ bits per qubit sent in
Ref.~\cite{leditzky-2018} and by about $5 \cdot 10^{-3}$ bits per
qubit sent in Ref.~\cite{siddhu-2020}. The difference between the
two-letter quantum capacity and the single-letter quantum capacity
was experimentally tested for the dephrasure channels in
Ref.~\cite{yu-2020}.

However, the physics of photon transmission through optical
communication lines is much richer and the losses are
polarization-dependent in general~\cite{gisin-1997}. This means
the transmission coefficient $p_H$ for a horizontally polarized
photon may significantly differ from the transmission coefficient
$p_V$ for a vertically polarized photon. Effect of polarization
dependent loss on the the quality of transmitted polarization
entanglement and the secure quantum communication is discussed in
Refs.~\cite{kirby-2019,li-2018}. The conventional erasure channel
is not an adequate description of polarization dependent losses.
Similarly, a concatenation of a quantum decoherence channel with
the erasure channel is not adequate in description of the
transmission of quantum carriers through general lossy
communication lines. In this paper, we fill this gap by
introducing a \emph{generalized erasure channel} that covers the
above phenomena. Our definition of the generalized erasure channel
differs from the generalized erasure channel pair considered in
Ref.~\cite{siddhu-2020}. In fact, our definition comprises all
concatenations of the erasure channel with other channels as
partial cases; however, our definition is applicable to a wider
class of scenarios with the state-dependent losses, which were not
considered before.

The key idea behind the generalized erasure channel is that
probabilistic transformations of quantum states are given by
quantum operations that are completely positive and trace
nonincreasing maps (see, e.g.,~\cite{heinosaari-2012}). Quantum
operations are extensively used in description of nondestructive
quantum measurements~\cite{davies-1970} and schemes with
sequential
measurements~\cite{carmeli-2011,luchnikov-2017,zhuravlev-2020,leppajarvi-2020}.
Importantly, quantum operations do not generally reduce to
attenuated quantum channels. These are \emph{biased} quantum
operations that exhibit the state-dependent probability to lose an
information carrier. In this paper, we study physics of the biased
quantum operations and relate it with the information transmission
through lossy communication lines.

We define the generalized erasure channel as an orthogonal sum of
a trace decreasing quantum operation and a map outputting the
state-dependent probability to lose the particle. This enables us
to treat any type of information capacity for a trace decreasing
quantum operation as the same type of capacity for the
corresponding generalized erasure channel. We focus on the
classical and quantum capacities of the generalized erasure
channels and derive lower and upper bounds for them. We elaborate
the physical scenario of the polarization dependent losses and
discover superadditivity of coherent information for the
corresponding generalized erasure channel. For a region of
transmission coefficients for horizontally and vertically
polarized photons we analytically prove that the two-letter
quantum capacity is strictly greater the single-letter quantum
capacity. The maximum difference exceeds $7.197 \cdot 10^{-3}$
bits per qubit sent, which is the greatest reported difference
among qubit-input channels to the best of our knowledge.

The paper is organized as follows. In
Sections~\ref{section-subnormalized-d-o}
and~\ref{section-quantum-operations}, we review subnormalized
density operators and trace nonincreasing quantum operations that
can be either unbiased or biased depending on the trace of their
output. In Section~\ref{section-extension}, we study the ways in
which quantum operations can be extended to trace preserving maps
and find the minimal extension such that all other extensions are
derivatives of the minimal one. In
Section~\ref{section-normalized-image}, we study the normalized
image of a trace decreasing quantum operation $\Lambda$ and show
that this image coincides with the image of some quantum channel
$\Phi_{\Lambda}$, which will be later used in estimation of bounds
for capacities. In Section~\ref{section-gec}, we give a precise
definition of the generalized erasure channel. In
Section~\ref{section-gec-C}, we find lower and upper bounds for
the classical capacity and the single-letter classical capacity of
the generalized erasure channel. In Section~\ref{section-gec-Q},
we (i) find lower and upper bounds for the quantum capacity and
the singe-letter quantum capacity of the generalized erasure
channel; (ii) calculate the singe-letter quantum capacity and
estimate the two-letter quantum capacity for a generalized erasure
channel describing the polarization dependent losses, providing a
proof for superadditivity of coherent information within a wide
range of polarization transmission coefficients $p_H$ and $p_V$.
In Seciton~\ref{section-conclusions}, brief conclusions are given.

\section{Trace decreasing quantum operations}

\subsection{Subnormalized density operators}
\label{section-subnormalized-d-o}

We consider $d$-level quantum systems as information carriers, $1
< d < \infty$. By ${\cal H}_d$ denote a $d$-dimensional Hilbert
space associated with a single system. ${\cal B}({\cal H}_d) =
\{X: {\cal H}_d \to {\cal H}_d \}$ is the set of linear operators
acting on ${\cal H}_d$. A quantum state of a single information
carrier is given by a density operator $\rho \in {\cal B}({\cal
H}_d)$ that is positive-semidefinite and has unit trace. By ${\cal
D}({\cal H}_d)$ denote the set of density operators on ${\cal
H}_d$, i.e., ${\cal D}({\cal H}_d) = \{ \rho \in {\cal B}({\cal
H}_d) \, | \, \rho^{\dag} = \rho \geq 0, \ {\rm tr}[\rho] = 1 \}$.
Any physical quantity $f$ associated with the information carrier
is mathematically described by a self-adjoint operator $F \in
{\cal B}({\cal H}_d)$ such that its mean value is given by the
Born rule $\ave{f} = {\rm tr}[\rho F]$.

For instance, if information is encoded into polarization degrees
of freedom of a single photon, then $d=2$ and ${\cal H}_2 = {\rm
Span}(\ket{H},\ket{V})$, where $\ket{H}$ and $\ket{V}$ are the
orthogonal state vectors describing a photon with horizontal and
vertical polarization, respectively. Let $f = \pm 1$ be the values
assigned to clicks of detectors located at two outputs of the
conventional polarization beam splitter (see Fig.~\ref{figure-1}).
A single photon in the state $\rho$ induces a click of one
detector only, with the probabilities being $p(f=1) = \bra{H} \rho
\ket{H}$ and $p(f=-1) = \bra{V} \rho \ket{V}$. The average
$\ave{f} = \sum_{f=\pm 1} f p(f) = {\rm tr}[\rho F]$, where $F =
\ket{H}\bra{H} - \ket{V}\bra{V}$.

\begin{figure}
\includegraphics[width=7cm]{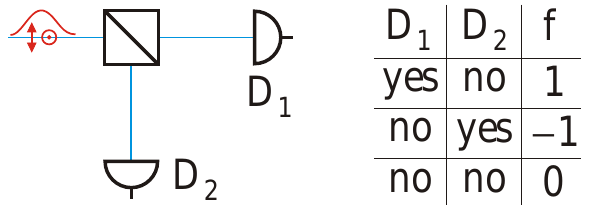}
\caption{\label{figure-1} Operational meaning of a subnormalized
density operator. The label ``yes'' (``no'') corresponds to a
detector click (no detector click).}
\end{figure}

The need to take a possible loss of the information carrier into
account can be satisfied as follows. Extending the Hilbert space
by a flag (vacuum) state $\ket{\rm vac}$, we can use the extended
density operator $R \in {\cal B}({\cal H}_{d+1})$. A measurement
of a physical quantity $f$ associated with the information carrier
would give no outcome at all if the carrier is lost, so we assign
the value $f=0$ if this situation takes place. For instance, if a
photon is lost, then none of the detectors at the outputs of the
polarization beam splitter clicks, which we interpret as the
outcome $0$ of quantity $f$ (see Fig.~\ref{figure-1}). The average
$\ave{f} = {\rm tr}[R (F \oplus 0 \ket{{\rm vac}}\bra{{\rm vac}} )
] = {\rm tr}[PRP^{\dag} \, F]$, where $P:{\cal H}_{d+1}
\rightarrow {\cal H}_d$ is a projector onto the original Hilbert
space associated with the information carrier. The operator
$\varrho = PRP^{\dag}$ is Hermitian, positive-semidefinite, and
its trace ${\rm tr}[PRP^{\dag}] \leq 1$, so we refer to it as a
subnormalized density operator. The trace of the subnormalized
density operator, ${\rm tr}[\varrho]$, is nothing else but the
probability to detect the information carrier. The probability to
lose the carrier equals $1 - {\rm tr}[\varrho]$. By ${\cal
S}({\cal H}_d)$ denote the set of subnormalized density operators
on ${\cal H}_d$, i.e.,
\begin{equation}
{\cal S}({\cal H}_d) = \{ \varrho \in {\cal B}({\cal H}_d) \, | \,
\varrho^{\dag} = \varrho \geq 0, \ {\rm tr}[\varrho] \leq 1 \}.
\end{equation}

\noindent The set of subnormalized density operators is a subset
of the cone of positive-semidefinite operators.

\subsection{Quantum operations} \label{section-quantum-operations}

A deterministic physical transformation of a density operator is
given by a quantum channel $\Phi:{\cal B}({\cal H}_d) \to {\cal
B}({\cal H}_d)$ that is a completely positive and trace preserving
linear map~\cite{holevo-2012}. Here the term ``deterministic''
refers to the unit probability to detect a particle, which implies
the trace preservation because any density operator is to be
mapped to a density operator. The trace preservation condition for
a linear map $\Phi$ takes a simpler form in terms of the dual map
$\Phi^{\dag}$ defined through formula ${\rm tr}\big[ \Phi[X] Y
\big] = {\rm tr}\big[ X \Phi^{\dag}[Y] \big]$ that is valid for
all $X,Y \in {\cal B}({\cal H}_d)$. As ${\rm
tr}\big[\Phi[\rho]\big] = {\rm tr}\big[\Phi[\rho] I \big] = {\rm
tr}\big[\rho \Phi^{\dag}[I] \big]$, $\Phi$ is trace preserving if
and only if $\Phi^{\dag}[I] =I$, where $I \in {\cal B}({\cal
H}_d)$ is the identity operator. Since the system in interest can
be potentially entangled with an auxiliary quantum system,
complete positivity guarantees that any density operator $\rho \in
{\cal D}({\cal H}_{d+k})$ for the aggregate of the system and the
auxiliary system of dimension $k$ is mapped to a valid density
operator $\Phi \otimes {\rm Id}_k [\rho]$, where ${\rm Id}_k:{\cal
B}({\cal H}_k) \to {\cal B}({\cal H}_k)$ is the identity map.
Technically, complete positivity of $\Phi$ means $\Phi \otimes
{\rm Id}_k [\rho] \geq 0$ for any $\rho \in {\cal D}({\cal
H}_{d+k})$ and any dimension $k \in \mathbb{N}$. By ${\cal
C}({\cal H}_d)$ denote the set of quantum channels for
$d$-dimensional systems, i.e., ${\cal C}({\cal H}_d) = \{
\Phi:{\cal B}({\cal H}_d) \to {\cal B}({\cal H}_d) \, | \, \Phi
\text{~is~completely~positive~and~} \ \Phi^{\dag}[I] =I \}$.
Informational properties of quantum channels including classical
and quantum capacity are reviewed in the
books~\cite{holevo-2012,wilde-2013}.

Since losses in a communication line diminish the probability to
detect the information carrier, a general physical transformation
$\Lambda: {\cal S}({\cal H}_d) \to {\cal S}({\cal H}_d)$ is trace
nonincreasing, i.e., ${\rm tr} \big[ \Lambda[\varrho] \big] \leq
{\rm tr}[\varrho]$ for all $\varrho \in {\cal S}({\cal H}_d)$. The
fact that $\Lambda$ is trace nonincreasing is equivalent to the
relation $\Lambda^{\dag}[I] \leq I$, where $A \leq B$ means $B-A$
is positive semidefinite. Complete positivity of a physical map
$\Lambda$ follows from the same line of reasoning as in the case
of quantum channels. Combining the two requirements, we get the
following definition of a general physical transformation (see,
e.g.,~\cite{heinosaari-ziman}). A linear map $\Lambda:{\cal
B}({\cal H}_d) \to {\cal B}({\cal H}_d)$ is called a quantum
operation if $\Lambda$ is completely positive and trace
nonincreasing. The concept of quantum operation is extensively
used, e.g., to describe state transformations induced by a general
nonprojective measurement (quantum instrument)~\cite{davies-1970}.
By ${\cal O}({\cal H}_d)$ denote the set of quantum operations for
a $d$-dimensional system, i.e.,
\begin{equation}
{\cal O}({\cal H}_d) = \{ \Lambda: {\cal B}({\cal H}_d) \to {\cal
B}({\cal H}_d) \, | \, \Lambda \text{~is~completely~positive~and~}
\Lambda^{\dag}[I] \leq I \}.
\end{equation}

If an operation $\Lambda \in {\cal O}({\cal H}_d) \setminus {\cal
C}({\cal H}_d)$, then $\Lambda$ is called trace decreasing. In
this paper, we focus on trace decreasing quantum operations and
their informational properties. There are two distinctive classes
of trace decreasing operations:

\begin{itemize}

\item[(i)] unbiased operations that are attenuated quantum
channels of the form $\Lambda = p\Phi$, where $0 \leq p \leq 1$
and $\Phi \in {\cal C}({\cal H}_d)$, for which the output
$\Lambda[\rho]$ is detected with a fixed probability ${\rm
tr}\big[ \Lambda[\rho] \big] = p$ regardless of the input density
operator $\rho$;

\item[(ii)] biased operations with the state-dependent probability
to detect the outcome, i.e., there exist density operators
$\rho_1$ and $\rho_2$ such that ${\rm tr} \big[ \Lambda[\rho_1]
\big] \neq {\rm tr} \big[ \Lambda[\rho_2] \big]$.

\end{itemize}

Biased operations are of primary interest in this paper. A
physical example of the biased operation is an optical fiber with
polarization dependent losses~\cite{gisin-1997}. Suppose the least
attenuated polarization state is either a horizontally polarized
state or a vertically polarized state, and let $p_H$ and $p_V$ be
the attenuation factors for the horizontal and vertical
polarizations, respectively. Then the effect of the optical fiber
with polarization dependent losses, in its simplest form, is
described by the following quantum operation with one Kraus
operator $A$~\cite{gisin-1997}:
\begin{equation} \label{pdl}
\Lambda[\varrho] = A \varrho A^{\dag}, \quad A = \sqrt{p_H}
\ket{H}\bra{H} + \sqrt{p_V} \ket{V}\bra{V}.
\end{equation}

We quantify the bias of a quantum operation $\Lambda \in {\cal
O}({\cal H}_d)$ by
\begin{equation} \label{bias}
b(\Lambda) = \sup_{\rho \in {\cal D}({\cal H}_d)}{\rm tr} \big[
\Lambda[\rho] \big] - \inf_{\rho \in {\cal D}({\cal H}_d)}{\rm tr}
\big[ \Lambda[\rho] \big].
\end{equation}

\noindent Clearly, the quantity~\eqref{bias} vanishes if and only
if the operation $\Lambda$ is unbiased. Using the formalism of
dual maps, we readily get
\begin{equation} \label{bias-simple}
b(\Lambda) = \max \, {\rm Spec} \left( \Lambda^{\dag}[I] \right) -
\min \, {\rm Spec} \left( \Lambda^{\dag}[I] \right),
\end{equation}

\noindent where ${\rm Spec} \left( X \right)$ is a spectrum of
$X$. For the operation~\eqref{pdl} we have $\Lambda^{\dag}[I] =
p_H \ket{H}\bra{H} + p_V \ket{V}\bra{V}$, so its bias $b(\Lambda)
=|p_H - p_V|$.

\subsection{Extending a trace decreasing operation to a
channel} \label{section-extension}

Any trace decreasing operation can be extended to a quantum
channel by adding another trace decreasing operation. A quantum
operation $\Lambda'$ is called an extension for a quantum
operation $\Lambda$ if $\Lambda + \Lambda'$ is trace preserving.
In terms of the dual maps the latter condition reads
$\Lambda^{\dag}[I] + (\Lambda')^{\dag}[I] = I$, which uniquely
defines the operator $(\Lambda')^{\dag}[I] = I -
\Lambda^{\dag}[I]$ but does not fix the map $\Lambda'$, so the
extension is not unique in general. In fact, since $\Phi^{\dag}[I]
= I$ for any quantum channel $\Phi$, then $(\Lambda')^{\dag}\left[
\Phi^{\dag}[I] \right] = (\Lambda')^{\dag}[I]$, i.e., a
concatenation $\Phi \circ \Lambda'$ is an extension for $\Lambda$
provided $\Lambda'$ is an extension for $\Lambda$ too. The set
$\{\Phi \circ \Lambda' | \Phi \in {\cal C}({\cal H}_d)\}$ is
called an orbit of the operation $\Lambda' \in {\cal O}({\cal
H}_d)$. We have just shown that any map from the orbit of some
extension $\Lambda'$ is an extension too, but a natural question
arises if all possible extensions can be obtained as an orbit of a
single (in a sense, minimal) extension $\Lambda'_{\min}$? The
following proposition answers this question in affirmative.

\begin{proposition} \label{proposition-extension}
The map
\begin{equation} \label{min-ext}
\Lambda'_{\min}[\varrho] = \sqrt{I - \Lambda^{\dag}[I]} \, \varrho
\, \sqrt{I - \Lambda^{\dag}[I]}
\end{equation}

\noindent is a minimal extension for the quantum operation
$\Lambda \in {\cal O}({\cal H}_d)$, i.e., any extension for
$\Lambda$ has the form $\Phi \circ \Lambda'_{\min}$, $\Phi \in
{\cal C}({\cal H}_d)$.
\end{proposition}

\begin{proof}
The map~\eqref{min-ext} is an extension for $\Lambda$ because it
is completely positive, trace nonincreasing, and
$(\Lambda'_{\min})^{\dag}[I] = I - \Lambda^{\dag}[I]$. Let $P_+
\in {\cal B}({\cal H}_d)$ denote the projector onto the support of
$I - \Lambda^{\dag}[I]$. By $P_0 \in {\cal B}({\cal H}_d)$ we
denote the projector onto the kernel of $I - \Lambda^{\dag}[I]$.
Let $X^{-1}$ be the Moore--Penrose inverse of $X$, then the
operator $(I - \Lambda^{\dag}[I])^{-1/2}$ is well defined and its
support coincides with the support of $P_+$.

Consider any other extension $\Lambda'$ for $\Lambda$ and its
Kraus decomposition $\Lambda'[\varrho] = \sum_l B_l \varrho
B_l^{\dag}$. Then $(\Lambda')^{\dag}[I] = \sum_{l} B_l^{\dag}B_l =
I - \Lambda^{\dag}[I]$ and ${\rm supp} B_l = {\rm supp}
B_l^{\dag}B_l \subset {\rm supp} P_+$. Define the completely
positive map $\Phi$ by the Kraus sum $\Phi[\varrho] = P_0 \varrho
P_0 + \sum_l B_l (I - \Lambda^{\dag}[I])^{-1/2} \varrho (I -
\Lambda^{\dag}[I])^{-1/2} B_l^{\dag}$, then $\Phi^{\dag}[I] = P_0
+ (I - \Lambda^{\dag}[I])^{-1/2} \sum_l B_l^{\dag} B_l (I -
\Lambda^{\dag}[I])^{-1/2} = P_0 + P_+ = I$ and $\Phi$ is trace
preserving, which implies $\Phi \in {\cal C}({\cal H}_d)$. On the
other hand, $\Phi \left[ \Lambda'_{\min}[\varrho] \right] = \sum_l
B_l P_+ \varrho P_+ B_l^{\dag}$. Recalling the relation ${\rm
supp} B_l = {\rm supp} B_l^{\dag}B_l \subset {\rm supp} P_+$, we
conclude that $\Phi \circ \Lambda'_{\min} = \Lambda'$.
\end{proof}

Note that the minimal extension~\eqref{min-ext} has the Kraus rank
1. If $\Lambda[\varrho] = A \varrho A^{\dag}$,
$\Lambda'_{\min}[\varrho] = B \varrho B^{\dag}$, where $B =
\sqrt{I - A^{\dag}A}$. In particular, for an optical fiber with
polarization dependent losses given by Eq.~\eqref{pdl} we have $B
= \sqrt{1-p_H} \ket{H}\bra{H} + \sqrt{1-p_V} \ket{V}\bra{V}$.

\subsection{Normalized image of a trace decreasing operation}
\label{section-normalized-image}

Consider a trace decreasing operation $\Lambda \in {\cal O}({\cal
H}_d)$ and its outcome $\Lambda[\rho]$ for some density operator
$\rho \in {\cal D}({\cal H}_d)$. Suppose the probability to detect
the outcome particle is nonzero, i.e., ${\rm tr}\left[ \rho
\Lambda^{\dag}[I] \right] \neq 0$. While measuring a physical
quantity $f$ in a statistical experiment described in
section~\ref{section-subnormalized-d-o}, we can exclude the
outcomes when none of the detectors clicks and get a conditional
distribution for values of $f$. This is equivalent to normalizing
the output operator $\Lambda[\rho]$, which leads to a map
$\Lambda_{\cal D}: {\cal D}({\cal H}_d) \to {\cal D}({\cal H}_d)$
defined by
\begin{equation} \label{non-lin}
\Lambda_{\cal D} [\rho] = \frac{\Lambda[\rho]}{{\rm tr}\big[
\Lambda[\rho] \big]}, \quad {\rm tr}\big[ \Lambda[\rho] \big] \neq
0.
\end{equation}

Eq.~\eqref{non-lin} describes a conditional output density
operator that is commonly reconstructed in quantum optics
experiments via postselection (see, e.g.,~\cite{bogdanov-2006}).
Eq.~\eqref{non-lin} is also used to describe a conditional output
state of a quantum measuring
apparatus~\cite{davies-1970,carmeli-2011,luchnikov-2017,zhuravlev-2020,leppajarvi-2020}.
If $\Lambda$ is unbiased, i.e., $\Lambda = p \Phi$ for some
$0<p\leq 1$ and some channel $\Phi$, then $\Lambda_{\cal D}
[\sum_i \lambda_i \rho_i] = \sum_i \lambda_i \Lambda_{\cal D}
[\rho_i]$ for any ensemble $\{\lambda_i, \rho_i\}$ of density
operators, $\{\lambda_i\}_i$ being the probability distribution.
In other words, for an unbiased operation $\Lambda$ the map
$\Lambda_{\cal D}$ is quasi-linear on convex sums of density
operators; however, $\Lambda_{\cal D}$ is nonlinear in general
because $\Lambda_{\cal D} [c\rho] = \Lambda_{\cal D} [\rho]$ for
all $c \neq 0$. If $\Lambda$ is biased, then quasi-linearity does
not hold and $\Lambda_{\cal D} [\sum_i \lambda_i \rho_i] \neq
\sum_i \lambda_i \Lambda_{\cal D} [\rho_i]$.

Our main interest in this section is the image $\Lambda_{\cal D}
[{\cal D}({\cal H}_d)] = \{ \Lambda_{\cal D}[\rho] | \rho \in
{\cal D}({\cal H}_d) \}$ that consists of all conditional output
density operators. As $\Lambda_{\cal D}$ is nonlinear and does not
exhibit quasi-linearity in general, one may expect that
$\Lambda_{\cal D} [{\cal D}({\cal H}_d)]$ significantly differs
from the image $\Phi [{\cal D}({\cal H}_d)]$ of any quantum
channel $\Phi \in {\cal C}({\cal H}_d)$. The following example
shows that this is not the case.

\begin{example} \label{example-Lambda-Phi-Lambda}
Consider a qubit operation $\Lambda:{\cal B}({\cal H}_2) \to {\cal
B}({\cal H}_2)$ of the form
\begin{equation*}
\Lambda[\varrho] = \frac{1}{2} \left( a {\rm tr}[\varrho] I + b
{\rm tr}[\varrho \sigma_x] \sigma_x + b {\rm tr}[\varrho \sigma_y]
\sigma_y + {\rm tr}[\varrho \sigma_z] (c \sigma_z + d I) \right),
\end{equation*}

\noindent where $\sigma_x,\sigma_y,\sigma_z$ is the conventional
set of Pauli operators and real parameters $a,b,c,d$ satisfy the
relations $a \geq |c|+|d|$, $(a+c)^2 \geq 4b^2 + d^2$, and $a+|d|
\leq 1$, which make $\Lambda$ be completely positive and trace
nonincreasing. Since $\Lambda^{\dag}[I] = a I + d\sigma_z$, the
bias $b(\Lambda) = 2|d|$. If $a = |d|$, then $b=c=0$ and the image
$\Lambda_{\cal D} [{\cal D}({\cal H}_d)]$ becomes highly
degenerate, namely, $\Lambda_{\cal D} [{\cal D}({\cal H}_d)
\setminus \{\rho_0\}] = \{\frac{1}{2}I\}$, where $\rho_0 =
\frac{1}{2}[I -{\rm sgn}(d) \sigma_z]$ has vanishing detection
probability, $\Lambda[\rho_0] = 0$. In the following, we consider
the case $a > |d|$.

The Bloch vector parametrization for $\rho \in {\cal D}({\cal
H}_2)$ reads $\rho = \frac{1}{2}(I+{\bf
r}\cdot\boldsymbol{\sigma})$, ${\bf r} \in \mathbb{R}^3$, $|{\bf
r}| \leq 1$. The conditional output density operator
$\Lambda_{\cal D}[\rho]$ has the Bloch vector ${\bf q}$ with
components $q_x = b r_x / (a + d r_z)$, $q_y = b r_y / (a + d
r_z)$, $q_z = c r_z / (a + d r_z)$. Rewriting the inequality ${\bf
r} \cdot {\bf r} \leq 1$ in terms of ${\bf q}$, we get
\begin{equation*}
\frac{q_x^2}{\left(\frac{b}{\sqrt{a^2-d^2}}\right)^2} +
\frac{q_y^2}{\left(\frac{b}{\sqrt{a^2-d^2}}\right)^2} +
\frac{\left(q_z +
\frac{cd}{a^2-d^2}\right)^2}{\left(\frac{ac}{a^2-d^2}\right)^2}
\leq 1,
\end{equation*}

\noindent which defines an ellipsoid of revolution in
$\mathbb{R}^3$. Not any ellipsoid within a unit ball can be
associated with the image of a quantum
channel~\cite{heinosaari-ziman}. In our case, the image
$\Lambda_{\cal D} [{\cal D}({\cal H}_d)]$ coincides with the image
$\Phi_{\Lambda}[{\cal D}({\cal H}_d)]$ of the phase covariant
map~\cite{fgl-2020}
\begin{equation} \label{Phi-Lambda-specific}
\Phi_{\Lambda} [\varrho] = \frac{1}{2} \left( {\rm tr}[\varrho] (I
+ t_z\sigma_z) + \lambda {\rm tr}[\varrho \sigma_x] \sigma_x +
\lambda {\rm tr}[\varrho \sigma_y] \sigma_y + \lambda_z {\rm
tr}[\varrho \sigma_z] \sigma_z \right), \quad \lambda =
\frac{b}{\sqrt{a^2-d^2}}, \lambda_z = \frac{ac}{a^2-d^2}, t_z = -
\frac{cd}{a^2-d^2}.
\end{equation}

\noindent The map~\eqref{Phi-Lambda-specific} is clearly trace
preserving, whereas it is completely positive if and only if
$|\lambda_z| + |t_z| \leq 1$ and $4 \lambda^2 + t_z^2 \leq (1 +
\lambda_z)^2$ (see Ref.~\cite{fgl-2020}), with the both conditions
being automatically fulfilled if $\Lambda$ is a valid quantum
operation and $a > |d|$. If $a=|d|$, then we put
$\lambda=\lambda_z=t_z=0$. Finally, $\Lambda_{\cal D} [{\cal
D}({\cal H}_d)] = \Phi_{\Lambda}[{\cal D}({\cal H}_d)]$, where
$\Phi_{\Lambda}$ is a quantum channel.
\end{example}

Generalizing the above example to arbitrary qubit operations, one
can readily see that the image of a nonlinear qubt
map~\eqref{non-lin} is the same as the image of some linear,
completely positive, and trace preserving qubit map
$\Phi_{\Lambda}$. The following result generalizes this relation
further for an arbitrary finite dimension $d$ of the underlying
Hilbert space ${\cal H}_d$ and specifies the explicit form of the
channel $\Phi_{\Lambda}$.

\begin{figure}
\includegraphics[width=7cm]{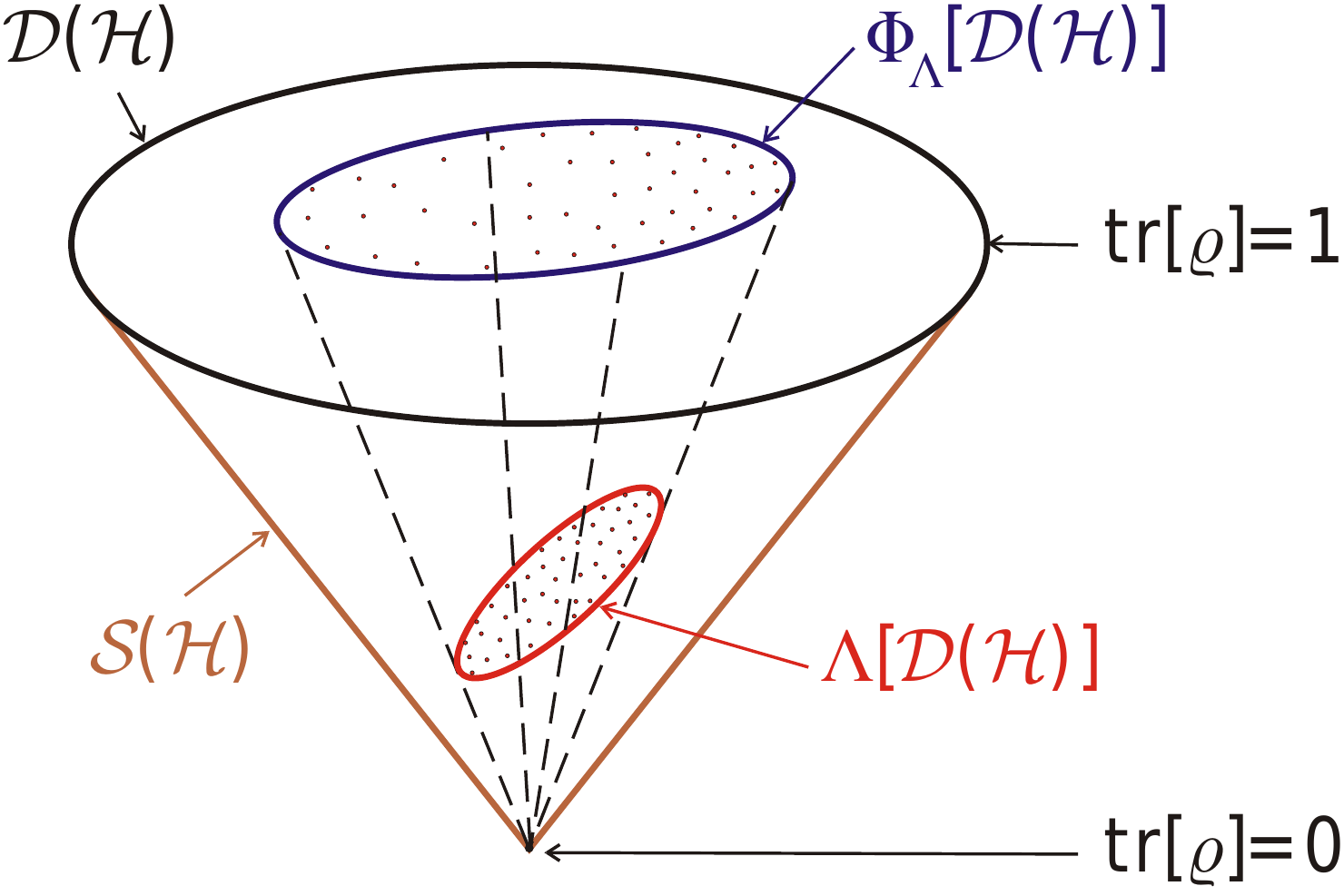}
\caption{\label{figure-2} Normalized image of a trace decreasing
operation $\Lambda$ coincides the image the channel
$\Phi_{\Lambda}$.}
\end{figure}

\begin{proposition} \label{proposition-images}
For a quantum operation $\Lambda \in {\cal O}({\cal H}_d)$,
$\Lambda^{\dag}[I] \neq 0$, the image $\Lambda_{\cal D}[{\cal
D}({\cal H}_d)]$ coincides with the image $\Phi_{\Lambda}[{\cal
D}({\cal H}_d)]$ of the quantum channel
\begin{equation} \label{Phi-Lambda}
\Phi_{\Lambda} [\rho] = \Lambda\left[ (\Lambda^{\dag}[I])^{-1/2}
\rho (\Lambda^{\dag}[I])^{-1/2} \right] + {\rm tr}[\varrho \Pi_0]
\xi,
\end{equation}

\noindent where $X^{-1}$ is the Moore--Penrose inverse of $X \in
{\cal B}({\cal H}_d)$, $\Pi_0 \in {\cal B}({\cal H}_d)$ is a
projector onto the kernel of operator $\Lambda^{\dag}[I]$, and
$\xi$ is an arbitrary density operator from the image
$\Lambda_{\cal D}[{\cal D}({\cal H}_d)]$.
\end{proposition}

\begin{proof}
We note that $\Phi$ is completely positive as the maps $X \to
(\Lambda^{\dag}[I])^{-1/2} X (\Lambda^{\dag}[I])^{-1/2}$, $X \to
\Pi_0 X \Pi_0$, $X \to {\rm tr}[X] \xi$, and $\Lambda$ are all
completely positive. Denoting by $\Pi_+ \in {\cal B}({\cal H}_d)$
the projector onto the support of operator $\Lambda^{\dag}[I]$, we
calculate $\Phi^{\dag}[I] = (\Lambda^{\dag}[I])^{-1/2}
\Lambda^{\dag}[I] (\Lambda^{\dag}[I])^{-1/2} + {\rm tr}[\xi] \Pi_0
= \Pi_+ + \Pi_0 = I$. Therefore, $\Phi$ is trace preserving and,
consequently, $\Phi$ is a quantum channel.

Given the Kraus decomposition $\Lambda[\rho] = \sum_k A_k \rho
A_k^{\dag}$, we have $\sum_{k} A_k^{\dag} A_k =
\Lambda^{\dag}[I]$, which implies ${\rm supp} A_k = {\rm supp}
A_k^{\dag} A_k \subset {\rm supp} \Lambda^{\dag}[I] = {\rm supp}
\Pi_+$. Then $\Lambda[\rho] = \Lambda[\Pi_+ \rho \Pi_+]$ for any
$\rho \in {\cal D}({\cal H}_d)$. Relying on the latter
observation, we divide ${\cal D}({\cal H}_d)$ into three subsets
and consider them separately.

(i) Consider a subset ${\cal D}_+({\cal H}_d) \subset {\cal
D}({\cal H}_d)$ of the density operators whose support belongs to
the support of $\Lambda^{\dag}[I]$, i.e., ${\cal D}_+({\cal H}_d)
= \{ \rho \in {\cal D}({\cal H}_d) \, | \, \Pi_+ \rho \Pi_+ = \rho
\}$. For any $\rho \in {\cal D}_+({\cal H}_d)$ we have ${\rm
tr}[\rho \Pi_0] = 0$ and
\begin{equation} \label{Phi-through-Lambda}
\Phi_{\Lambda}[\rho] = \Lambda\left[ (\Lambda^{\dag}[I])^{-1/2}
\rho (\Lambda^{\dag}[I])^{-1/2} \right] =
\frac{\Lambda[\rho']}{{\rm tr}\big[ \Lambda[\rho'] \big]}, \quad
\rho' = \frac{(\Lambda^{\dag}[I])^{-1/2} \rho
(\Lambda^{\dag}[I])^{-1/2}}{{\rm tr}\left[
(\Lambda^{\dag}[I])^{-1} \rho \right]} \in {\cal D}_+({\cal H}_d),
\end{equation}

\noindent i.e., for any $\rho \in {\cal D}_+({\cal H}_d)$ there
exists $\rho' \in {\cal D}_+({\cal H}_d)$ such that
$\Phi_{\Lambda}[\rho] = \Lambda_{\cal D}[\rho']$. Conversely, for
any $\rho' \in {\cal D}_+({\cal H}_d)$ we have
\begin{equation} \label{Lambda-through-Phi}
\Lambda_{\cal D}[\rho'] = \frac{\Lambda[\rho']}{{\rm tr}\big[
\Lambda[\rho'] \big]} = \Phi_{\Lambda}[\rho], \quad \rho =
\frac{\sqrt{\Lambda^{\dag}[I]} \, \rho' \,
\sqrt{\Lambda^{\dag}[I]}}{{\rm tr}\big[ \Lambda[\rho'] \big]} \in
{\cal D}_+({\cal H}_d),
\end{equation}

\noindent i.e., for any $\rho' \in {\cal D}_+({\cal H}_d)$ there
exists $\rho \in {\cal D}_+({\cal H}_d)$ such that $\Lambda_{\cal
D}[\rho'] = \Phi_{\Lambda}[\rho]$. Recalling the fact that
$\Lambda[\rho] = \Lambda[\Pi_+ \rho \Pi_+]$ for any $\rho \in
{\cal D}({\cal H}_d)$ and combining it with
Eqs.~\eqref{Phi-through-Lambda} and~\eqref{Lambda-through-Phi}, we
conclude
\begin{equation} \label{set-relations-Lambda-Phi}
\Lambda_{\cal D}[{\cal D}({\cal H}_d)] = \Lambda_{\cal D}[{\cal
D}_+({\cal H}_d)] = \Phi_{\Lambda}[{\cal D}_+({\cal H}_d)].
\end{equation}

(ii) Suppose ${\rm supp} \rho \subset {\rm supp} \Pi_0$, then
${\rm tr}[\varrho \Pi_0] = 1$ and $\Phi_{\Lambda}[\rho] = \xi \in
\Lambda_{\cal D}[{\cal D}({\cal H}_d)]$ by the statement of
proposition. Hence, $\Phi_{\Lambda}[\rho] \in \Phi_{\Lambda}[{\cal
D}_+({\cal H}_d)]$.

(iii) Consider the case when $\rho \in {\cal D}({\cal H}_d)
\setminus {\cal D}_+({\cal H}_d)$ but ${\rm supp} \rho \not\subset
{\rm supp} \Pi_0$, then $p_+:={\rm tr}[\Pi_+ \rho \Pi_+] > 0$,
$p_0:={\rm tr}[\Pi_0 \rho \Pi_0] > 0$, $p_+ + p_0 =1$, and
\begin{equation*}
\Phi_{\Lambda}[\rho] = \Phi_{\Lambda}[\Pi_+ \rho \Pi_+] + {\rm
tr}[\varrho \Pi_0] \xi = p_+ \Lambda_{\cal D}[\rho'] + p_0 \xi \in
\Phi_{\Lambda}[{\cal D}_+({\cal H}_d)]
\end{equation*}

\noindent because $\Phi_{\Lambda}[{\cal D}_+({\cal H}_d)]$ is a
convex set.

Therefore, in all the considered cases we have
$\Phi_{\Lambda}[{\cal D}({\cal H}_d)] = \Phi_{\Lambda}[{\cal
D}_+({\cal H}_d)]$. Recalling
Eq.~\eqref{set-relations-Lambda-Phi}, we obtain the equality
$\Lambda_{\cal D}[{\cal D}({\cal H}_d)] = \Phi_{\Lambda}[{\cal
D}({\cal H}_d)]$.
\end{proof}

Though the images $\Lambda_{\cal D}[{\cal D}({\cal H}_d)]$ and
$\Phi_{\Lambda}[{\cal D}({\cal H}_d)]$ coincide, the distributions
of the output states for the maps $\Lambda_{\cal D}$ and
$\Phi_{\Lambda}$ are different in general. Provided the
distribution of the input density operators is uniform with
respect to the Hilbert-Schmidt measure~\cite{bz}, the outcome
distribution for $\Lambda_{\cal D}$ would be uniform only if
$\Lambda$ is unbiased. Fig.~\ref{figure-2} illustrates the
geometric meaning of the fact the higher the detection probability
${\rm tr}\big[\Lambda[\rho]\big]$ the greater the density of the
output states for the map $\Lambda_{\cal D}$. Fig.~\ref{figure-2}
explains the geometric meaning of
Proposition~\ref{proposition-images} too.

As a quantum operation $\Lambda$ and the corresponding quantum
channel $\Phi_{\Lambda}$ are intimately related, a natural
question arises if $\Lambda$ can be extended to $\Phi_{\Lambda}$.
The following example shows this is not the case in general.

\begin{example}
Let $\Lambda$ be the qubit operation~\eqref{pdl} describing
polarization dependent losses with $p_H > 0$ and $p_V > 0$. Then
$\Lambda^{\dag}[I] = p_H \ket{H}\bra{H} + p_V \ket{V}\bra{V}$ is
strictly positive and $\Phi_{\Lambda} = {\rm Id}$. Consider the
density operator $\rho = \frac{1}{2}(\ket{H}\bra{H} +
\ket{H}\bra{V} + \ket{V}\bra{H} + \ket{V}\bra{V})$, then the
operator $(\Phi_{\Lambda} - \Lambda)[\rho] =
\frac{1}{2}[(1-p_H)\ket{H}\bra{H} + (1-\sqrt{p_H
p_V})(\ket{H}\bra{V} + \ket{V}\bra{H}) + (1-p_V)\ket{V}\bra{V}]$
is not positive semidefinite whenever $p_H \neq p_V$, which
implies that the map $\Phi_{\Lambda} - \Lambda$ is not positive,
so $\Phi_{\Lambda} - \Lambda$ is not a quantum operation and
cannot be an extension for $\Lambda$.
\end{example}

\section{Generalized erasure channel} \label{section-gec}

A trace decreasing quantum operation $\Lambda \in {\cal O}({\cal
H}_d) \setminus {\cal C}({\cal H}_d)$ probabilistically describes
the information transmission through a lossy quantum communication
line. The probabilistic nature of that transmission is due to a
finite detection probability ${\rm tr}\big[\Lambda[\rho]\big] \leq
1$ of a single information carrier initially prepared in the state
$\rho \in {\cal D}({\cal H}_d)$. If the information carriers enter
the communication line within predefined time bins, then the loss
of a carrier is detected by recording no measurement outcome
within a given time bin, see Fig.~\ref{figure-3}. Operationally,
we treat the absence of a measurement outcome as the creation of
an erasure flag state $\ket{e}\bra{e}$ so that $\ket{e}$ is
orthogonal to any vector from ${\cal H}_d$. Therefore, we extend
the outcome Hilbert space to ${\cal H}_{d+1} = {\rm Span} \left(
{\cal H}_d \cup \ket{e} \right)$. However, as we conditionally
create the erasure state $\ket{e}\bra{e}$ upon detecting no
measurement outcome, there can be no coherent superposition
between a vector from ${\cal H}_d$ and the vector $\ket{e}$ in the
outcome. The probability to get the erasure state equals the
probability to observe no measurement outcome, ${\rm tr}\big[ \rho
- \Lambda[\rho] \big] = {\rm tr}\left[\varrho (I -
\Lambda^{\dag}[I]) \right]$. Therefore, the resulting
transformation of the input density operator $\rho \in {\cal
D}({\cal H}_d)$ is given by the following linear map ${\cal
B}({\cal H}_d) \to {\cal B}({\cal H}_{d+1})$:
\begin{equation} \label{gec}
\Gamma_{\Lambda}[\rho] = \Lambda[\rho] \oplus {\rm tr} \big[ \rho
-
\Lambda[\rho] \big] \ket{e}\bra{e} = \left(%
\begin{array}{cc}
  \Lambda[\rho] & {\bf 0} \\
  {\bf 0}^\top & {\rm tr}\left[\rho (I - \Lambda^{\dag}[I]) \right] \\
\end{array}%
\right),
\end{equation}

\noindent where the matrix form assumes the use of some
orthonormal basis $\{\ldots,\ket{e}\}$ in ${\cal H}_{d+1}$, ${\bf
0}$ is a $d$-component zero vector. Recalling the minimal
extension~\eqref{min-ext}, we see that $\Gamma_{\Lambda} = \Lambda
\oplus ({\rm Tr} \circ \Lambda'_{\min})$, where ${\rm Tr}[\rho] =
{\rm tr}[\rho] \ket{e}\bra{e}$ is a so-called trash-and-prepare
quantum channel~\cite{leppajarvi-2020}. Since both $\Lambda$ and
$\Lambda'_{\min}$ are completely positive, so is
$\Gamma_{\Lambda}$. The map~\eqref{gec} is trace preserving
because the map $\Lambda \oplus \Lambda'_{\min}$ is trace
preserving. Therefore, $\Gamma_{\Lambda}$ is a quantum channel. It
is worth mentioning that our definition is very similar to a map
considered in Ref.~\cite{heinosaari-2012}, section VI.B, where the
authors discuss the existence of a trace nonincreasing map
$\Lambda$ such that $\Lambda[\rho_i] = p_i \rho'_i$, $0 \leq p_i
\leq 1$, for two given sets of density operators
$\{\rho_i\}_{i=1}^N$ and $\{\rho'_i\}_{i=1}^N$.

\begin{figure}
\includegraphics[width=10cm]{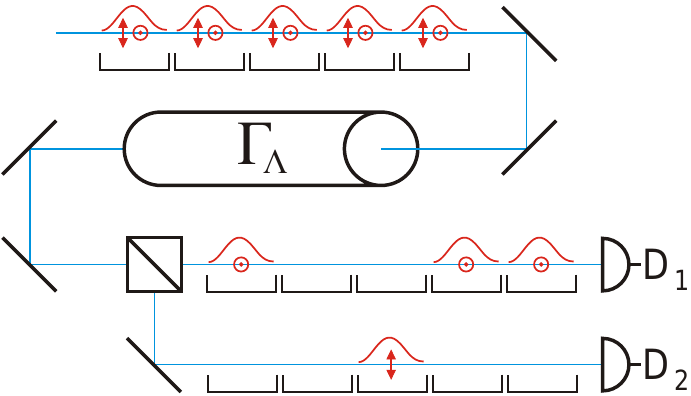}
\caption{\label{figure-3} Operational meaning of the generalized
erasure channel $\Gamma_{\Lambda}$, where the trace decreasing
quantum operation $\Lambda$ describes the polarization dependent
losses.}
\end{figure}

If $\Lambda = p {\rm Id}$, $0\leq p \leq 1$, then $\Gamma_{p{\rm
Id}}$ is nothing else but the conventional erasure
channel~\cite{grassl-1997,bennett-1997}. If $\Lambda = p \Phi$,
where $\Phi:{\cal B}({\cal H}_2) \to {\cal B}({\cal H}_2)$ is a
dephasure channel, then $\Gamma_{p \Phi}$ is a so-called
dephrasure channel~\cite{leditzky-2018}. The authors of
Ref.~\cite{siddhu-2020} consider $\Lambda = p \Phi$, where
$\Phi:{\cal B}({\cal H}_2) \to {\cal B}({\cal H}_2)$ is a general
channel or an amplitude damping channel, in particular. In
contrast to these specific cases, $\Lambda$ does not have to be
unbiased, so the erasure probability ${\rm tr}\left[\varrho (I -
\Lambda^{\dag}[I]) \right]$ is state-dependent in general. This is
the reason we refer to the channel~\eqref{gec} as a
\emph{generalized erasure channel}. Note that our definition
differs from the concept of the generalized erasure channel pair
introduced in Ref.~\cite{siddhu-2020}. As the information
transmission through a lossy communication line has operational
meaning only in the described scenario with predefined time bins,
we associate an information transmission capacity of the quantum
operation $\Lambda$ with the corresponding capacity of the
generalized erasure channel $\Gamma_{\Lambda}$. In the following
sections, we consider specific scenarios of classical and quantum
information transmission trough a lossy quantum communication
line.

\subsection{Classical capacity} \label{section-gec-C}

Encoding classical information into $d$-dimensional quantum
systems, sending all the systems through the same memoryless
quantum channel $\Phi$, and measuring the outcome, it becomes
possible to transmit classical information via quantum
communication lines. The maximum rate of reliable information
transmission per system used is called classical capacity and
reads~\cite{holevo-1998,schumacher-1997}
\begin{equation}
C(\Phi) = \lim_{n \rightarrow \infty} \frac{1}{n}
C_{\chi}(\Phi^{\otimes n}), \quad C_{\chi}(\Psi) =
\sup_{\{\pi_i,\rho_i\}} \bigg\{ S \bigg( \Psi \Big[ \sum_i \pi_i
\rho_i \Big] \bigg) - \sum_i \pi_i S \left( \Psi[\rho_i] \right)
\bigg\},
\end{equation}

\noindent where $C_{\chi}(\Psi)$ is the Holevo capacity of a
channel $\Psi:{\cal B}({\cal H}_{d'}) \to {\cal B}({\cal
H}_{d'})$, $S(\rho) = - {\rm tr}[\rho \log\rho]$ is the von
Neumann entropy, and $\{\pi_i,\rho_i\}$ is an ensemble of density
operators ($\pi_i \geq 0$, $\sum_i \pi_i = 1$, $\rho_i \in {\cal
D}({\cal H}_{d'})$). Hereafter, the base of the $\log$ can be
chosen at wish depending on the preferred units of information;
the base equals 2 if the information is quantified in bits. The
regularized capacity $C(\Phi)$ may exceed the Holevo capacity
$C_{\chi}(\Phi)$~\cite{hastings-2009}; however, it is hard to
evaluate $C(\Phi)$ explicitly for a given channel $\Phi$, so many
recent studies are devoted to the search of lower and upper bounds
for $C(\Phi)$ (see, e.g.,
Refs.~\cite{leditzky-wilde-2018,filippov-2018,fk-2019}). Below in
this section, we find the lower and upper bounds for classical
capacity of the generalized erasure channel.

As the concatenation $\Phi \circ \Psi$ of quantum channels $\Phi$
and $\Psi$ satisfies $C(\Phi \circ \Psi) \leq C(\Psi)$ (see,
e.g.,~\cite{holevo-2012}), we first establish an analogous
relation for generalized erasure channels.

\begin{proposition} \label{proposition-relation}
Suppose the quantum operations $\Lambda_1,\Lambda_2,\Theta \in
{\cal O}({\cal H}_d)$ satisfy the relation $\Lambda_1 = \Theta
\circ \Lambda_2$, then $C(\Gamma_{\Lambda_1}) \leq
C(\Gamma_{\Lambda_2})$.
\end{proposition}

\begin{proof}
Define the map $\Xi:{\cal B}({\cal H}_{d+1}) \to {\cal B}({\cal
H}_{d+1})$ by its action on matrices in the basis
$\{\ldots,\ket{e}\}$:
\begin{equation} \label{Xi-channel}
\Xi \left[ \left(%
\begin{array}{cc}
  \rho & \vdots \\
  \cdots & c \\
\end{array}%
\right) \right] = \left(%
\begin{array}{cc}
  \Theta[\rho] & {\bf 0} \\
  {\bf 0}^\top & c + {\rm tr}\left[\rho (I - \Theta^{\dag}[I]) \right] \\
\end{array}%
\right),
\end{equation}

\noindent then $\Xi$ is completely positive and trace preserving,
i.e., $\Xi \in {\cal C}({\cal H}_{d+1})$. It is not hard to see
that $\Gamma_{\Lambda_1} = \Xi \circ \Gamma_{\Lambda_2}$, which
implies $C(\Gamma_{\Lambda_1}) \leq C(\Gamma_{\Lambda_2})$ by the
concatenation property for quantum channels.
\end{proof}

Using the result of Proposition~\ref{proposition-relation}, we can
find an upper bound for $C(\Gamma_{\Lambda})$ in terms of the
quantum operation $\Lambda$.

\begin{proposition} \label{proposition-C-upper-bound}
Let $\Lambda \in {\cal O}({\cal H}_d)$, then $C(\Gamma_{\Lambda})
\leq (\log d) \max {\rm Spec} \left( \Lambda^{\dag}[I] \right)$.
\end{proposition}
\begin{proof}
If $\Lambda=0$, then apparently $C(\Gamma_{\Lambda}) = 0$. Suppose
$\Lambda \neq 0$, then $p_{\max} := \max {\rm Spec} \left(
\Lambda^{\dag}[I] \right) > 0$ and $\Theta = p_{\max}^{-1}
\Lambda$ is a valid quantum operation because $\Theta$ is
completely positive and $\Theta^{\dag}[I] \leq p_{\max}^{-1}
\Lambda^{\dag}[I] \leq I$. Therefore, $\Lambda = \Theta \circ
\Lambda_2$, where $\Lambda_2 = p_{\max} {\rm Id}$. By
Proposition~\ref{proposition-relation} we have
$C(\Gamma_{\Lambda}) \leq C(\Gamma_{\Lambda_2})$. On the other
hand, $\Gamma_{\Lambda_2} \equiv \Gamma_{p_{\max} {\rm Id}}$ is
the conventional erasure channel whose classical capacity is well
known~\cite{bennett-1997,amosov-2009}, namely, $C(\Gamma_{p_{\max}
{\rm Id}}) = p_{\max} \log d$.
\end{proof}

It is tempting to treat $(\log d) \min {\rm Spec} \left(
\Lambda^{\dag}[I] \right)$ as a lower bound for
$C(\Gamma_{\Lambda})$, but a simple counterexample is the unbiased
operation $\Lambda[\rho] = p \, {\rm tr}[\rho] \frac{1}{d}I$, for
which $C(\Gamma_{\Lambda}) = 0 < p \log d = (\log d) \min {\rm
Spec} \left( \Lambda^{\dag}[I] \right)$ if $0<p \leq 1$. As the
classical capacity $C(\Gamma_{\Lambda})$ may vanish for unbiased
quantum operations $\Lambda$, we need to establish a reasonable
lower bound for $C(\Gamma_{\Lambda})$ in the case of biased
quantum operations $\Lambda$.

\begin{proposition} \label{proposition-C-lower-bound}
Let $\Lambda \in {\cal O}({\cal H}_d)$, then $C(\Gamma_{\Lambda})
\geq F(p_{\min},p_{\max})$, where
\begin{equation} \label{lower-bound-for-C}
F(p_{\min},p_{\max}) = \left\{ \begin{array}{ll}
  \log \left( 1 + {\rm Exp}\left[ -
\frac{h(p_{\max})-h(p_{\min})}{p_{\max}-p_{\min}} \right] \right)
- \frac{p_{\max} h(p_{\min}) - p_{\min} h(p_{\max})}{p_{\max} -
p_{\min}} & \text{if~} p_{\min} < p_{\max},\\
  0 & \text{if~} p_{\min} = p_{\max}, \\
\end{array} \right.
\end{equation}

\noindent ${\rm Exp}$ is the inverse function to $\log$, $p_{\max}
= \max {\rm Spec} \left( \Lambda^{\dag}[I] \right)$, $p_{\min} =
\min {\rm Spec} \left( \Lambda^{\dag}[I] \right)$, and
\begin{equation*}
h(x) = - x \log x - (1-x) \log (1-x).
\end{equation*}
\end{proposition}

\begin{proof}
Consider a quantum channel $\Xi:{\cal B}({\cal H}_{d+1}) \to {\cal
B}({\cal H}_2)$, which affects matrices in the basis
$\{\ldots,\ket{e}\}$ as follows:
\begin{equation*}
\Xi \left[ \left(%
\begin{array}{cc}
  \rho & \vdots \\
  \cdots & c \\
\end{array}%
\right) \right] = \left(%
\begin{array}{cc}
  {\rm tr}[\rho] & 0 \\
  0 & c \\
\end{array}%
\right).
\end{equation*}

\noindent Then $\Xi \circ \Gamma_{\Lambda}$ is a quantum channel
too and
\begin{equation} \label{inequalities-in-lower-bound}
C(\Gamma_{\Lambda}) \geq C(\Xi \circ \Gamma_{\Lambda}) \geq
C_{\chi} (\Xi \circ \Gamma_{\Lambda}) \geq S \bigg( \Xi \circ
\Gamma_{\Lambda} \Big[ \sum_{i=1,2} \pi_i \rho_i \Big] \bigg) -
\sum_{i=1,2} \pi_i S \left( \Xi \circ \Gamma_{\Lambda} [\rho_i]
\right)
\end{equation}

\noindent for some ensemble $\{\pi_i,\rho_i\}_{i=1,2}$ consisting
of two density matrices $\rho_1,\rho_2 \in {\cal D}({\cal H}_d)$
emerging with probabilities $\pi_1$ and $\pi_2 = 1 - \pi_1$,
respectively. Let $\rho_1 = \ket{f_{\max}}\bra{f_{\max}}$ and
$\rho_2 = \ket{f_{\min}}\bra{f_{\min}}$, where $\ket{f_{\max}} \in
{\cal H}_d$ and $\ket{f_{\min}} \in {\cal H}_d$ are normalized
vectors such that $\Lambda^{\dag}[I] \ket{f_{\max}} = p_{\max}
\ket{f_{\max}}$ and $\Lambda^{\dag}[I] \ket{f_{\min}} = p_{\min}
\ket{f_{\min}}$. Then $\Xi \circ \Gamma_{\Lambda}[\rho_1] = {\rm
diag}(p_{\max},1-p_{\max})$ and $\Xi \circ
\Gamma_{\Lambda}[\rho_2] = {\rm diag}(p_{\min},1-p_{\min})$. The
rightmost side of Eq.~\eqref{inequalities-in-lower-bound} equals
$h(\pi_1 p_{\max} + \pi_2 p_{\min}) - \pi_1 h(p_{\max}) - \pi_2
h(p_{\min})$. Maximizing the latter expression with respect to a
binary probability distribution $(\pi_1,\pi_2)$, we get the right
hand side of Eq.~\eqref{lower-bound-for-C}.
\end{proof}

The lower bound~\eqref{lower-bound-for-C} is nonvanishing whenever
$p_{\max} > p_{\min}$, i.e., for any biased operation $\Lambda$.
The following result provides a lower bound for the classical
capacity of $\Gamma_{\Lambda}$ only in terms of the bias
$b(\Lambda)$.

\begin{corollary} Let $\Lambda \in {\cal O}({\cal H}_d)$, then
$C(\Gamma_{\Lambda}) \geq \log 2 - h \left( \frac{1+b(\Lambda)}{2}
\right)$, where $b(\Lambda) = p_{\max} - p_{\min}$, $p_{\max} =
{\rm Spec} \left( \Lambda^{\dag}[I] \right)$, $p_{\min} = {\rm
Spec} \left( \Lambda^{\dag}[I] \right)$.
\end{corollary}
\begin{proof}
If $b(\Lambda) = 0$, then we get a trivial bound
$C(\Gamma_{\Lambda}) \geq 0$. Suppose $b(\Lambda) > 0$ is fixed.
Let us use the expressions $p_{\min} = \frac{1}{2} + x -
\frac{1}{2} b(\Lambda)$ and $p_{\max} = \frac{1}{2} + x +
\frac{1}{2} b(\Lambda)$, where $x := \frac{1}{2} (p_{\min} +
p_{\max} - 1) \in [\frac{1}{2} b(\Lambda) - \frac{1}{2},
\frac{1}{2} - \frac{1}{2} b(\Lambda)]$. Then the lower bound
$F(p_{\min},p_{\max})$ for $C(\Gamma_{\Lambda})$ can be rewritten
as $f(x):= F(\frac{1}{2} + x - \frac{1}{2} b(\Lambda),\frac{1}{2}
+ x + \frac{1}{2} b(\Lambda))$. Note that $f(x) = f(-x)$ because
the replacement $x \rightarrow -x$ leads to $p_{\min} \rightarrow
1 - p_{\max}$ and $p_{\max} \rightarrow 1 - p_{\min}$, whereas
$\log(1 + {\rm Exp}[-y]) = \log(1 + {\rm Exp}[y])-y$. This means
$f(x)$ is an even function and
$\left.\frac{df}{dx}\right\vert_{x=0} = 0$. Moreover, $\left.
\frac{d^2 f}{dx^2} \right\vert_{x=0} = \left[
\frac{4b^2(\Lambda)}{1-b^2(\Lambda)} + \frac{1}{b^2(\Lambda)}
\left( {\rm ln} \frac{1-b(\Lambda)}{1+b(\Lambda)} + 2b(\Lambda)
\right)^2 \right] \log {\rm e}
> 0$, $\frac{df}{dx} < 0$ if $x < 0$, and $\frac{df}{dx} > 0$ if $x
> 0$. Therefore, $C(\Gamma_{\Lambda}) \geq F(p_{\min},p_{\max})
\geq f(0) = \log 2 - h \left( \frac{1+b(\Lambda)}{2} \right)$.
\end{proof}

\begin{example}
Consider the reflection of photons from a dielectric surface,
where the angle of incidence equals Brewster's angle. In this
case, we deal with the quantum operation $\Lambda \in {\cal
O}({\cal H}_2)$ given by Eq.~\eqref{pdl}, where $p_H > p_V = 0$.
Then $p_{\max} = p_H$, $p_{\min}=0$, and
Propositions~\ref{proposition-C-upper-bound}
and~\ref{proposition-C-lower-bound} yield $\log \left( 1 + p_H
(1-p_H)^{(1-p_H)/p_H} \right) \leq C(\Gamma_{\Lambda}) \leq p_H
\log 2$. If $p_H = 1$, then $C(\Gamma_{\Lambda}) = \log 2$.
\end{example}

The disadvantage of Propositions~\ref{proposition-C-upper-bound}
and~\ref{proposition-C-lower-bound} is that they exploit only two
quantities, $p_{\max}$ and $p_{\min}$, leaving the structure of
the quantum operation $\Lambda$ beyond the scope. As we know from
Proposition~\ref{proposition-images}, the normalized image of
$\Lambda$ coincides with the image of the channel $\Phi_{\Lambda}$
given by Eq.~\eqref{Phi-Lambda}, which enables us to relate the
Holevo capacity $C_{\chi}(\Gamma_{\Lambda})$ with the Holevo
capacity $C_{\chi}(\Phi_{\Lambda})$.

\begin{proposition} \label{proposition-C1-lower-upper-bounds}
Let $\Lambda \in {\cal O}({\cal H}_d)$, then $p_{\min}
C_{\chi}(\Phi_{\Lambda}) \leq C_{\chi}(\Gamma_{\Lambda}) \leq
p_{\max} C_{\chi}(\Phi_{\Lambda}) + F(p_{\min},p_{\max})$, where
$p_{\max} = \max {\rm Spec} \left( \Lambda^{\dag}[I] \right)$,
$p_{\min} = \min {\rm Spec} \left( \Lambda^{\dag}[I] \right)$,
$\Phi_{\Lambda}$ is given by Eq.~\eqref{Phi-Lambda}, and
$F(p_{\min},p_{\max})$ is given by Eq.~\eqref{lower-bound-for-C}.
\end{proposition}

\begin{proof}
Consider an ensemble $\{\pi_k,\xi_k\}$, where $\{\pi_k\}$ is a
nondegerate probability distribution and $\xi_k \in {\cal D}({\cal
H}_d)$. The Holevo capacity of a channel $\Psi$ reads
$C_{\chi}(\Psi) = \sup_{\{\pi_k,\xi_k\}}
\chi(\{\pi_k,\Psi[\xi_k]\})$, where $\chi(\{\pi_k,\Psi[\xi_k]\})
:= S(\Psi[\sum_k \pi_k \xi_k]) - \sum_k \pi_k S(\Psi[\xi_k])$ is
the so-called Holevo quantity. Let the ensemble $\{\pi_k,\xi_k\}$
pass through the generalized erasure channel $\Gamma_{\Lambda}$,
then the output ensemble is $\{\pi_k,\Gamma_{\Lambda}[\xi_k]\}$.
Since $\Lambda$ is trace nonincreasing, we have $\Lambda[\xi_k] =
p_k \rho_k$, where $p_k = {\rm tr}\big[\Lambda[\xi_k]\big] \in
[0,1]$ and $\rho_k \in {\cal D}({\cal H}_d)$. We have $S(p_k
\rho_k) = p_k S(\rho_k) - p_k \log p_k$, and a straightforward
calculation yields $S \left[ \left(%
\begin{array}{cc}
  p_k \rho_k & {\bf 0} \\
  {\bf 0}^{\top} & 1-p_k \\
\end{array}%
\right) \right] = p_k S(\rho_k) + h(p_k)$. Denote $\overline{p} =
\sum_k \pi_k p_k > 0$ and introduce the renormalized probabilities
$q_k = \pi_k p_k / \overline{p}$ and the average state
$\overline{\rho} = \sum_k q_k \rho_k$, then $\Lambda[ \sum_k \pi_k
\xi_k] = \sum_k \pi_k p_k \rho_k = \overline{p} \,
\overline{\rho}$. We obtain
\begin{eqnarray*}
&& \chi(\{\pi_k,\Gamma_{\Lambda}[\xi_k]\}) = S \Big(
\Gamma_{\Lambda} \Big[\sum_k \pi_k \xi_k \Big] \Big) - \sum_k
\pi_k S(\Gamma_{\Lambda}[\xi_k]) = \overline{p} S(\overline{\rho})
+ h(\overline{p}) - \sum_k \pi_k \big(
p_k S(\rho_k) + h(p_k) \big) \nonumber\\
&& = \overline{p} \left( S(\overline{\rho}) - \sum_k q_k S(\rho_k)
\right) + h\left( \sum_k \pi_k p_k \right) - \sum_k \pi_k h(p_k) =
\overline{p} \, \chi(\{q_k,\rho_k\}) + h\left( \sum_k \pi_k p_k
\right) - \sum_k \pi_k h(p_k).
\end{eqnarray*}

\noindent As for any $q_k >0$ we have $\rho_k =
\Phi_{\Lambda}[\rho'_k]$ for some $\rho'_k \in {\cal D}({\cal
H}_d)$ due to Proposition~\ref{proposition-images}, we obtain
\begin{equation} \label{chi-Gamma-Lambda-chi-Phi-Lambda}
\chi(\{\pi_k,\Gamma_{\Lambda}[\xi_k]\}) = \overline{p} \,
\chi(\{q_k, \Phi_{\Lambda}[\rho_k']\}) + h\left( \sum_k \pi_k p_k
\right) - \sum_k \pi_k h(p_k).
\end{equation}

Let us consider two cases.

(i) Suppose $\{\pi_k,\xi_k\}$ is an optimal ensemble such that
$C_{\chi}(\Gamma_{\Lambda}) =
\chi(\{\pi_k,\Gamma_{\Lambda}[\xi_k]\})$, then $\{q_k,\rho_k'\}$
is some (generally nonoptimal) ensemble and $\chi(\{q_k,
\Phi_{\Lambda}[\rho_k']\}) \leq C_{\chi}(\Phi_{\Lambda})$. As
$p_{\min} \leq p_k \leq p_{\max}$ for all $k$, we have
$\overline{p} \leq p_{\max}$. Also, for any $k$ there exists
$\mu_k \in [0,1]$ such that $p_k = \mu_k p_{\min} + (1-\mu_k)
p_{\max}$. Since the binary entropy $h(x)$ is a concave function,
$h(p_k) \geq \mu_k h(p_{\min}) + (1-\mu_k) h(p_{\max})$. Denote
$\widetilde{\pi}_1 := \sum_k \pi_k \mu_k$ and $\widetilde{\pi}_2
:= \sum_k \pi_k (1 - \mu_k)$, then
$(\widetilde{\pi}_1,\widetilde{\pi}_2)$ is a binary probability
distribution and $\sum_k \pi_k h(p_k) \geq \widetilde{\pi}_1
h(p_{\min}) + \widetilde{\pi}_2 h(p_{\max})$, which implies
$h\left( \sum_k \pi_k p_k \right) - \sum_k \pi_k h(p_k) \leq
h(\widetilde{\pi}_1 p_{\min} + \widetilde{\pi}_2 p_{\max}) -
\widetilde{\pi}_1 h(p_{\min}) - \widetilde{\pi}_2 h(p_{\max}) \leq
F(p_{\min},p_{\max})$. Finally, we get the upper bound
$C_{\chi}(\Gamma_{\Lambda}) \leq p_{\max} C_{\chi}(\Phi_{\Lambda})
+ F(p_{\min},p_{\max})$.

(ii) Suppose $\{q_k,\rho'_k\}$ is an optimal ensemble such that
$C_{\chi}(\Phi_{\Lambda}) = \chi(\{q_k,
\Phi_{\Lambda}[\rho_k']\})$. For any $\rho'_k$ there exists $\xi_k
\in {\cal D}({\cal H}_d)$ such that $\Lambda[\xi_k] / p_k =
\Phi_{\Lambda}[\rho_k']$ and $p_k = {\rm
tr}\big[\Lambda[\xi_k]\big] > 0$ due to
Proposition~\ref{proposition-images}. Define $1 / \overline{p} =
\sum_k q_k / p_k$, then formula $\pi_k = \overline{p}q_k / p_k$
defines a probability distribution $\{\pi_k\}$ and
Eq.~\eqref{chi-Gamma-Lambda-chi-Phi-Lambda} is valid. Since
$\{\pi_k,\xi_k\}$ is some (generally nonoptimal) ensemble for the
channel $\Gamma_{\Lambda}$, we have $C_{\chi}(\Phi_{\Lambda}) \geq
\chi(\{\pi_k,\Gamma_{\Lambda}[\xi_k]\}) = \overline{p}
C_{\chi}(\Phi_{\Lambda}) + h\left( \sum_k \pi_k p_k \right) -
\sum_k \pi_k h(p_k) \geq p_{\min} C_{\chi}(\Phi_{\Lambda})$.
\end{proof}

Note that the inequality $p_{\min} C_{\chi}(\Phi_{\Lambda}) \leq
C_{\chi}(\Gamma_{\Lambda})$ cannot be derived in a way similar to
the proof of Proposition~\ref{proposition-relation} because, in
general, there exists no quantum operation $\Theta$ such that
$\Theta \circ \Lambda = p_{\min} \Phi_{\Lambda}$. For instance,
for the operation $\Lambda$ in
Example~\ref{example-Lambda-Phi-Lambda} we explicitly find the
unique $\Theta = p_{\min} \Phi_{\Lambda} \circ \Lambda^{-1}$ if
$abc \neq 0$; however, the obtained map $\Theta$ turns out to be
nonpositive.

For unbiased operations $\Lambda$ we have $p_{\min} = p_{\max}$
and $F(p_{\min},p_{\max}) = 0$, so we readily get the following
result.

\begin{corollary}
Let $\Lambda \in {\cal O}({\cal H}_d)$ be an unbiased quantum
operation, i.e., $\Lambda = p \Phi$ for some $0 \leq p \leq 1$ and
a quantum channel $\Phi \in {\cal C}({\cal H}_d)$, then
$C_{\chi}(\Gamma_{\Lambda}) = p C_{\chi}(\Phi)$.
\end{corollary}

To conclude this section, we establish the relation between tensor
products $\Gamma_{\Lambda}^{\otimes n}$ and $\Lambda^{\otimes n}$.
Using the definition~\eqref{gec} and the expression
$\Gamma_{\Lambda} = \Lambda \oplus ({\rm Tr} \circ
\Lambda'_{\min})$, it becomes clear that
\begin{equation} \label{Gamma-Lambda-tensor-power-2}
\Gamma_{\Lambda}^{\otimes 2}  = \Lambda^{\otimes 2} \ \oplus \
\left\{ \Lambda \otimes ({\rm Tr} \circ \Lambda'_{\min}) \right\}
\ \oplus \ \left\{ ({\rm Tr} \circ \Lambda'_{\min}) \otimes
\Lambda \right\} \ \oplus \ ({\rm Tr} \circ
\Lambda'_{\min})^{\otimes 2}.
\end{equation}

\noindent Similarly, we have $\Gamma_{\Lambda}^{\otimes n} =
\Lambda^{\otimes n} \ \oplus \ \Upsilon$, where the image of the
map $\Upsilon$ is orthogonal to the image of the map
$\Lambda^{\otimes n}$ with respect to the Hilbert--Schmidt scalar
product. For any $\rho \in {\cal B}({\cal H}_{d}^{\otimes n})$ the
support of $\Upsilon[\rho]$ belongs to a linear subspace of
dimension $(d+1)^n - d^n$; we denote this subspace by ${\cal
H}_{d+1}^{\otimes n} \setminus {\cal H}_{d}^{\otimes n}$. Let
${\rm Id}^{\otimes n}: {\cal B}({\cal H}_{d}^{\otimes n}) \to
{\cal B}({\cal H}_{d}^{\otimes n})$ be the identity transformation
and ${\rm Tr}_{{\cal H}_{d+1}^{\otimes n} \setminus {\cal
H}_{d}^{\otimes n}}$ be a trash-and-prepare quantum channel that
maps any $\varrho \in {\cal B}({\cal H}_{d+1}^{\otimes n}
\setminus {\cal H}_{d}^{\otimes n})$ to ${\rm tr}[\varrho]
\ket{e}\bra{e}$. Since both $\Gamma_{\Lambda}^{\otimes n}$ and
${\rm Tr}_{{\cal H}_{d+1}^{\otimes n} \setminus {\cal
H}_{d}^{\otimes n}}$ are trace preserving, we have $( {\rm
Id}^{\otimes n} \oplus {\rm Tr}_{{\cal H}_{d+1}^{\otimes n}
\setminus {\cal H}_{d}^{\otimes n}} ) \circ
\Gamma_{\Lambda}^{\otimes n} [\rho] = \Lambda^{\otimes n}[\rho]
\oplus {\rm tr} \big[ \rho - \Lambda^{\otimes n}[\rho] \big]
\ket{e}\bra{e} = \Gamma_{\Lambda^{\otimes n}}[\rho]$. Therefore,
the channel $\Gamma_{\Lambda^{\otimes n}}$ is a concatenation of
channels $\Gamma_{\Lambda}^{\otimes n}$ and $( {\rm Id}^{\otimes
n} \oplus {\rm Tr}_{{\cal H}_{d+1}^{\otimes n} \setminus {\cal
H}_{d}^{\otimes n}} )$, and we immediately get the following
result.
\begin{proposition}
Let $\Lambda \in {\cal O}({\cal H}_d)$, then $C(\Gamma_{\Lambda})
\geq \frac{1}{n} C_{\chi} (\Gamma_{\Lambda}^{\otimes n}) \geq
\frac{1}{n} C_{\chi} (\Gamma_{\Lambda^{\otimes n}})$ for all $n
\in \mathbb{N}$.
\end{proposition}

\subsection{Quantum capacity} \label{section-gec-Q}

Encoding quantum states into higher-dimensional multipartite
quantum systems via an isometric map, sending all the systems
through the same memoryless quantum channel $\Phi$, and decoding
the outcome via a dimension-reducing quantum channel, one can
asymptotically achieve the perfect transfer of any initial quantum
state provided the noise in the communication line is not too
intense~\cite{lloyd-1997,barnum-1998,devetak-2005}. The rate of
this quantum communication is quantified by the logarithm of the
transferred state dimension per channel use. The maximum reliable
communication rate is called quantum capacity of the channel
$\Phi$ and reads~\cite{devetak-2005}
\begin{equation}
Q(\Phi) = \lim_{n \rightarrow \infty} \frac{1}{n}
Q_1(\Phi^{\otimes n}), \quad Q_1(\Psi) = \sup_{\rho \in {\cal
D}({\cal H}_{d'})} \{ S(\Psi[\rho]) - S(\widetilde{\Psi}[\rho])
\},
\end{equation}

\noindent where $\widetilde{\Psi}:{\cal B}({\cal H}_{d'}) \to
{\cal B}({\cal H}_{k})$ is a complementary channel to the channel
$\Psi:{\cal B}({\cal H}_{d'}) \to {\cal B}({\cal H}_{d''})$ with
the Kraus rank $k$. To be precise, $\widetilde{\Psi}[\rho] = {\rm
tr}_{{\cal H}_{d''}}[W \rho W^{\dag}]$, where $W:{\cal H}_{d'} \to
{\cal H}_{d''} \otimes {\cal H}_{k}$ is an isometry ($W^{\dag}W =
I$) in the Stinespring dilation $\Psi[\rho] = {\rm tr}_{{\cal
H}_{k}}[W \rho W^{\dag}]$. Physically, the complementary channel
output $\widetilde{\Psi}[\rho] \in {\cal D}({\cal H}_k)$ shows an
effective state of the environment after a density operator $\rho
\in {\cal D}({\cal H}_{d'})$ has passed through a quantum channel
$\Psi:{\cal B}({\cal H}_{d'}) \to {\cal B}({\cal H}_{d''})$ with
the Kraus rank $k$. The quantity $S(\Psi[\rho]) -
S(\widetilde{\Psi}[\rho])$ is known as the coherent information,
whereas $\frac{1}{n}Q_1(\Phi^{\otimes n})$ is usually referred to
as an $n$-letter quantum capacity. Useful conditions for strict
positivity of $Q_1(\Phi)$ are given in
Ref.~\cite{siddhu-mar-2020}.

The quantum capacity is known to satisfy the additivity property
$Q(\Phi) = Q_1(\Phi)$ if $\Phi$ is degradable, i.e., if there
exists a quantum channel $\Xi$ such that $\widetilde{\Phi} = \Xi
\circ \Phi$~\cite{devetak-shor-2005}. If $\Phi$ is antidegradable,
i.e., there exists a quantum channel $\Xi$ such that $\Phi = \Xi
\circ \widetilde{\Phi}$, then $Q(\Phi) = 0$ and the additivity
property is trivially fulfilled (see, e.g.,~\cite{cubitt-2008}).
The superadditivity of coherent information, i.e., the strict
inequality $Q_1(\Phi^{\otimes n})
> n Q_1(\Phi)$, is known to hold for some depolarizing channels if $n \geq 3$~\cite{divincenzo-1998,fern-2008}, some
dephrasure channels if $n \geq 2$~\cite{leditzky-2018},
concatenations of the erasure channel with the amplitude damping
channels~\cite{siddhu-2020}, the state-of-the-art channels
$\Phi:{\cal B}({\cal H}_3) \to {\cal B}({\cal H}_3)$ with
$\frac{1}{2}Q_1(\Phi^{\otimes 2}) - Q_1(\Phi) \approx 4.4 \cdot
10^{-2}$ and their higher-dimensional
generalizations~\cite{siddhu-nov-2020}, and for a collection of
peculiar channels if $n \geq n_0$, where $n_0 \geq 2$ specifies
the channel and can be arbitrary~\cite{cubitt-2015}. In this
section, we find lower and upper bounds for the quantum capacity
of a generalized erasure channel. Then we study degradability and
antidegradability for a class of generalized erasure channels. For
a 2-parameter map of that class, we reveal the superadditivity
property $Q_1(\Gamma_{\Lambda}^{\otimes 2}) > 2
Q_1(\Gamma_{\Lambda})$ within a wide range of parameters.

\begin{proposition}
Let $\Lambda \in {\cal O}({\cal H}_d)$, then $Q(\Gamma_{\Lambda})
\geq \frac{1}{n} Q_{1} (\Gamma_{\Lambda}^{\otimes n}) \geq
\frac{1}{n} Q_{1} (\Gamma_{\Lambda^{\otimes n}})$ for all $n \in
\mathbb{N}$.
\end{proposition}
\begin{proof}
The proof readily follows from the relation
$\Gamma_{\Lambda^{\otimes n}} = {\rm tr}_{{\cal H}_{d+1}^{\otimes
n} \setminus {\cal H}_{d}^{\otimes n}} \circ
\Gamma_{\Lambda}^{\otimes n}$ and the property $Q(\Psi_2 \circ
\Psi_1) \leq Q(\Psi_1)$ for concatenated quantum channels (see,
e.g.,~\cite{holevo-2012}).
\end{proof}

\begin{proposition} \label{proposition-relation-Q}
Suppose the quantum operations $\Lambda_1,\Lambda_2,\Theta \in
{\cal O}({\cal H}_d)$ satisfy the relation $\Lambda_1 = \Theta
\circ \Lambda_2$, then $Q(\Gamma_{\Lambda_1}) \leq
Q(\Gamma_{\Lambda_2})$.
\end{proposition}

\begin{proof}
Following the lines of Proposition~\ref{proposition-relation}, we
get $\Gamma_{\Lambda_1} = \Xi \circ \Gamma_{\Lambda_2}$, where the
channel $\Xi$ is given by Eq.~\eqref{Xi-channel}. By the
concatenation property for quantum channels we have
$Q(\Gamma_{\Lambda_1}) \leq Q(\Gamma_{\Lambda_2})$.
\end{proof}

\begin{proposition} \label{proposition-Q-upper-bound}
Let $\Lambda \in {\cal O}({\cal H}_d)$, then $Q(\Gamma_{\Lambda})
\leq \max(0,2 p_{\max} - 1) \log d$, where $p_{\max} = \max {\rm
Spec} \left( \Lambda^{\dag}[I] \right)$.
\end{proposition}
\begin{proof}
If $\Lambda=0$, then apparently $Q(\Gamma_{\Lambda}) = 0$. Suppose
$\Lambda \neq 0$, then $p_{\max} > 0$ and $\Theta = p_{\max}^{-1}
\Lambda$ is a valid quantum operation because $\Theta$ is
completely positive and $\Theta^{\dag}[I] \leq p_{\max}^{-1}
\Lambda^{\dag}[I] \leq I$. Therefore, $\Lambda = \Theta \circ
\Lambda_2$, where $\Lambda_2 = p_{\max} {\rm Id}$. By
Proposition~\ref{proposition-relation-Q} we have
$Q(\Gamma_{\Lambda}) \leq Q(\Gamma_{p_{\max} {\rm Id}})$. On the
other hand, $Q(\Gamma_{p_{\max} {\rm Id}}) = \max(0,2p_{\max} - 1)
\log d$, see Ref.~\cite{bennett-1997,holevo-2012}.
\end{proof}

The relation between the operation $\Lambda$ and the channel
$\Phi_{\Lambda}$ in Eq.~\eqref{Phi-Lambda} enables us to find both
the lower and upper bounds for $Q_1(\Gamma_{\Lambda})$ in terms of
$Q_1(\Phi_{\Lambda})$.

\begin{proposition} \label{proposition-Q1-lower-upper-bound}
Let $\Lambda \in {\cal O}({\cal H}_d)$ and suppose
$\Lambda^{\dag}[I]
> 0$, then
\begin{equation} \label{lower-and-upper-bounds-for-Q1}
p_{\min} Q_1(\Phi_{\Lambda}) - (1-p_{\min}) \log d \leq
Q_1(\Gamma_{\Lambda}) \leq p_{\max} Q_1(\Phi_{\Lambda}),
\end{equation}

\noindent where $p_{\max} = \max {\rm Spec} \left(
\Lambda^{\dag}[I] \right)$, $p_{\min} = \min {\rm Spec} \left(
\Lambda^{\dag}[I] \right)$, and $\Phi_{\Lambda}$ given by
Eq.~\eqref{Phi-Lambda}.
\end{proposition}
\begin{proof}
Since $\Lambda^{\dag}[I]
> 0$, the operator $\Pi_0$ in Eq.~\eqref{Phi-Lambda} is the zero operator and $\Phi_{\Lambda} [\rho] = \Lambda\left[
(\Lambda^{\dag}[I])^{-1/2} \rho (\Lambda^{\dag}[I])^{-1/2}
\right]$. We can rewrite the generalized erasure channel
$\Gamma_{\Lambda}$ in the form
\begin{equation} \label{direct-through-Phi-Lambda}
\Gamma_{\Lambda}[\rho] = \left(%
\begin{array}{cc}
\Phi_{\Lambda}\left[ \sqrt{\Lambda^{\dag}[I]} \, \rho \,
\sqrt{\Lambda^{\dag}[I]}
\right] & {\bf 0} \\
  {\bf 0}^\top & {\rm tr}\left[ \sqrt{I - \Lambda^{\dag}[I]} \, \rho \, \sqrt{I - \Lambda^{\dag}[I]} \right] \\
\end{array}%
\right).
\end{equation}

Let $\Phi_{\Lambda}[\rho] = \sum_{\alpha} V_{\alpha} \rho
V_{\alpha}^{\dag}$, then the Kraus operators $\widetilde{V}_{j}$
of the complementary channel $\widetilde{\Phi_{\Lambda}}$ satisfy
$\bra{\alpha} \widetilde{V}_{j} = \bra{j} V_{\alpha}$, where
$\{\ket{j}\}$ is an orthonormal basis for the output Hilbert space
and $\{\ket{\alpha}\}$ is an orthonormal basis for the effective
environment~\cite{holevo-2007}. Hence, $\widetilde{V}_{j} =
\sum_{\alpha} \ket{\alpha}\bra{j} V_{\alpha}$ and
$\widetilde{V}_{j} \sqrt{\Lambda^{\dag}[I]} = \sum_{\alpha}
\ket{\alpha}\bra{j} V_{\alpha} \sqrt{\Lambda^{\dag}[I]}$, i.e.,
the map $\rho \to \widetilde{\Phi_{\Lambda}} \left[
\sqrt{\Lambda^{\dag}[I]} \, \rho \, \sqrt{\Lambda^{\dag}[I]}
\right]$ is complementary to the map $\rho \to
\Phi_{\Lambda}\left[ \sqrt{\Lambda^{\dag}[I]} \, \rho \,
\sqrt{\Lambda^{\dag}[I]} \right]$. Since the identity channel
${\rm Id}$ is known to be complementary to the trash-and-prepare
channel ${\rm Tr}$ (see, e.g.,~\cite{holevo-2012}), we conclude
that the map $\rho \to \sqrt{I - \Lambda^{\dag}[I]} \, \rho \,
\sqrt{I - \Lambda^{\dag}[I]}$ is complementary to the map $\rho
\to {\rm tr}\left[ \sqrt{I - \Lambda^{\dag}[I]} \, \rho \, \sqrt{I
- \Lambda^{\dag}[I]} \right]$. Therefore,
\begin{equation} \label{complementary-through-Phi-Lambda}
\widetilde{\Gamma_{\Lambda}}[\rho] = \left(%
\begin{array}{cc}
 \widetilde{\Phi_{\Lambda}} \left[ \sqrt{\Lambda^{\dag}[I]} \, \rho \,
\sqrt{\Lambda^{\dag}[I]}
\right] & O \\
  O &  \sqrt{I - \Lambda^{\dag}[I]} \, \rho \, \sqrt{I - \Lambda^{\dag}[I]} \\
\end{array}%
\right).
\end{equation}

\noindent Let $\rho \in {\cal D}({\cal H}_d)$. Denoting $p = {\rm
tr} \big[ \rho \Lambda^{\dag}[I] \big] \in (0,1]$, $\xi = p^{-1}
\sqrt{\Lambda^{\dag}[I]} \, \rho \, \sqrt{\Lambda^{\dag}[I]} \in
{\cal D}({\cal H}_d)$, and $\omega = (1-p)^{-1} \sqrt{I -
\Lambda^{\dag}[I]} \, \rho \, \sqrt{I - \Lambda^{\dag}[I]} \in
{\cal D}({\cal H}_d)$ if $p \neq 1$, with $\omega \in {\cal
D}({\cal H}_d)$ being arbitrary if $p = 1$, we get
\begin{eqnarray} \label{Q1-Gamma-Lambda-Q1-Phi-Lambda}
&& S(\Gamma_{\Lambda}[\rho]) -
S(\widetilde{\Gamma_{\Lambda}}[\rho])
= S\left[ \left(%
\begin{array}{cc}
  p \Phi_{\Lambda}[\xi] & {\bf 0} \\
  {\bf 0}^{\top} & 1-p \\
\end{array}%
\right) \right] - S\left[ \left(%
\begin{array}{cc}
  p \widetilde{\Phi_{\Lambda}}[\xi] & O \\
  O & (1-p) \omega \\
\end{array}%
\right) \right] \nonumber\\
&& = p S(\Phi_{\Lambda}[\xi]) - p S
(\widetilde{\Phi_{\Lambda}}[\xi]) - (1-p) S(\omega).
\end{eqnarray}

Let us consider two cases.

(i) Suppose $\rho$ is optimal in the sense that
$Q_1(\Gamma_{\Lambda}) = S(\Gamma_{\Lambda}[\rho]) -
S(\widetilde{\Gamma_{\Lambda}}[\rho])$, then
Eq.~\eqref{Q1-Gamma-Lambda-Q1-Phi-Lambda} implies
$Q_1(\Gamma_{\Lambda}) \leq p S(\Phi_{\Lambda}[\xi]) - p S
(\widetilde{\Phi_{\Lambda}}[\xi]) \leq p_{\max}
Q_1(\Phi_{\Lambda})$.

(ii) Suppose $\xi$ is optimal in the sense that
$Q_1(\Phi_{\Lambda}) = S(\Phi_{\Lambda}[\xi]) -
S(\widetilde{\Phi_{\Lambda}}[\xi])$, then
Eq.~\eqref{Q1-Gamma-Lambda-Q1-Phi-Lambda} implies
$Q_1(\Gamma_{\Lambda}) \geq S(\Gamma_{\Lambda}[\rho]) -
S(\widetilde{\Gamma_{\Lambda}}[\rho]) = p Q_1(\Phi_{\Lambda}) -
(1-p) S(\omega) \geq p_{\min} Q_1(\Phi_{\Lambda}) - (1-p_{\min})
\log d$.
\end{proof}

\begin{example} \label{example-Q1}
Let $\Lambda \in {\cal O}({\cal H}_2)$ be a quantum operation
describing polarization dependent losses, Eq.~\eqref{pdl}. If $p_H
p_V \neq 0$, then $\Phi_{\Lambda} = {\rm Id}$ and
Proposition~\ref{proposition-Q1-lower-upper-bound} yields $(2
\min(p_H,p_V) - 1) \log 2 \leq Q_1(\Gamma_{\Lambda}) \leq
\max(p_H,p_V) \log 2$. Proposition~\ref{proposition-Q-upper-bound}
gives a tighter upper bound, namely, $Q_1(\Gamma_{\Lambda}) \leq
Q(\Gamma_{\Lambda}) \leq [2 \max(p_H,p_V) - 1] \log 2$. Hence,
$Q_1(\Gamma_{\Lambda}) = 0$ if $\max(p_H,p_V) \leq \frac{1}{2}$.
Suppose $\max(p_H,p_V) > \frac{1}{2}$ and $p_H p_V \neq 0$. We fix
the orthonormal basis $\{\ket{H}, \ket{V}\}$ and consider a
general input density matrix $\rho = \left(%
\begin{array}{cc}
  \rho_{HH} & \rho_{HV} \\
  \rho_{VH} & \rho_{VV} \\
\end{array}%
\right)$, $\rho_{HH} + \rho_{VV} = 1$. Then $\Lambda[\rho] = \left(%
\begin{array}{cc}
   p_H \rho_{HH} & \sqrt{p_H p_V} \rho_{HV} \\
  \sqrt{p_H p_V} \rho_{VH} & p_V \rho_{VV} \\
\end{array}%
\right)$ and the probability to detect a photon at the output
equals ${\rm tr}\big[ \Lambda[\rho] \big] = p_H \rho_{HH} + p_V
\rho_{VV}
> 0$. Since $\Phi_{\Lambda} = {\rm Id}$ and
$\widetilde{\Phi_{\Lambda}} = {\rm Tr}$,
Eqs.~\eqref{direct-through-Phi-Lambda} and
\eqref{complementary-through-Phi-Lambda} take the form
\begin{eqnarray*}
&& \Gamma_{\Lambda}[\rho] = \left(%
\begin{array}{ccc}
   p_H \rho_{HH} & \sqrt{p_H p_V} \rho_{HV} & 0 \\
  \sqrt{p_H p_V} \rho_{VH} & p_V \rho_{VV} & 0 \\
  0 & 0 & 1 - p_H \rho_{HH} - p_V \rho_{VV}
\end{array}%
\right),\\
&& \widetilde{\Gamma_{\Lambda}}[\rho] = \left(%
\begin{array}{ccc}
  p_H \rho_{HH} + p_V \rho_{VV} & 0 & 0 \\
  0 & (1-p_H) \rho_{HH} & \sqrt{(1-p_H)(1-p_V)} \rho_{HV} \\
  0 & \sqrt{(1-p_H)(1-p_V)} \rho_{VH} & (1-p_V) \rho_{VV}
\end{array}%
\right).
\end{eqnarray*}

\noindent Consider the function
$q(\rho):=S(\Gamma_{\Lambda}[\rho]) -
S(\widetilde{\Gamma_{\Lambda}}[\rho])$, then
$Q_1(\Gamma_{\Lambda}) = \max_{\rho \in {\cal D}({\cal H}_2)}
q(\rho)$. To find the maximum of $q(\rho)$, we first notice that
the entropies $S(\Gamma_{\Lambda}[\rho])$ and
$S(\widetilde{\Gamma_{\Lambda}}[\rho])$ do not depend on the phase
of $\rho_{HV}$, so we put $\rho_{HV} = \rho_{VH} = \frac{x}{2}
\geq 0$ and use the Bloch ball parametrization $\rho =
\frac{1}{2}(I + x \sigma_x + z \sigma_z)$, where $0 \leq x \leq
\sqrt{1-z^2}$. At the boundary $x = \sqrt{1-z^2}$ the function $q$
vanishes because this boundary corresponds to pure states for
which $S(\Gamma_{\Lambda}[\ket{\psi}\bra{\psi}]) =
S(\widetilde{\Gamma_{\Lambda}}[\ket{\psi}\bra{\psi}])$ (see,
e.g.,~\cite{holevo-2012}). Suppose $\max_{\rho \in {\cal D}({\cal
H}_2)} q(\rho) > 0$. Then this maximum is attained at some point
$(x_{\ast},z_{\ast})$ satisfying $0 \leq x_{\ast} <
\sqrt{1-z_{\ast}^2}$. Note that $z_{\ast} \in
(-\sqrt{1-x_{\ast}^2},\sqrt{1-x_{\ast}^2})$, so with necessity
$\left. \frac{\partial q}{\partial z} \right\vert_{x=x_{\ast},z =
z_{\ast}} = 0$. Consider an interior point $(x',z_{\ast})$, where
$0 < x' < \sqrt{1-z_{\ast}^2}$, $q \vert_{x=x',z = z_{\ast}} > 0$,
and $\frac{\partial q}{\partial z} \vert_{x=x',z = z_{\ast}} = 0$.
The direct calculation yields $\frac{\partial q}{\partial
x}\vert_{x=x',z = z_{\ast}} < 0$, which means $x_{\ast} \neq x'$
and the maximum cannot be attained at the interior point, so with
necessity $x_{\ast} = 0$. To find $z_{\ast}$ we need to solve the
equation $\frac{d}{dz} q \left( \frac{1}{2} (I + z \sigma_z)
\right) = 0$. We simplify
\begin{eqnarray} \label{coh-inf-simplified}
q \left( \frac{I + z \sigma_z}{2} \right) & = & {\rm H} \big( \{
\tfrac{1}{2} p_H (1+z), \tfrac{1}{2} p_V (1-z), 1-\tfrac{1}{2} p_H
(1+z)- \tfrac{1}{2} p_V (1-z) \} \big) \nonumber\\
&& - {\rm H} \big( \{ \tfrac{1}{2} p_H (1+z) + \tfrac{1}{2} p_V
(1-z), \tfrac{1}{2} (1-p_H) (1+z), \tfrac{1}{2} (1-p_V) (1-z) \}
\big),
\end{eqnarray}

\noindent where ${\rm H}(\{\lambda_i\}_i)$ is the Shannon entropy
of the probability distribution $\{\lambda_i\}_i$. It is not hard
to see that the equation $\frac{d}{dz} q \left( \frac{1}{2} (I + z
\sigma_z) \right) = 0$ is equivalent to the equation $G(p_H,p_V,z)
= G(p_V,p_H,-z)$, where
\begin{equation*}
G(p_1,p_2,z) = - p_1 \log \frac{p_1(1+z)}{p_1(1+z)+p_2(1-z)} +
(1-p_1) \log \frac{(1-p_1)(1+z)}{(1-p_1)(1+z)+(1-p_2)(1-z)}.
\end{equation*}

\noindent The analysis of derivative $\frac{d}{dz} q \left(
\frac{1}{2} (I + z \sigma_z) \right)$ shows that the maximum of
$q$ corresponds to such a solution $z=z_{\ast}$ of the equation
$G(p_H,p_V,z) = G(p_V,p_H,-z)$ for which ${\rm sgn} (z_{\ast}) =
{\rm sgn} (p_V - p_H)$. Substituting this solution into
Eq.~\eqref{coh-inf-simplified} it can be readily checked that $q
\left( \frac{1}{2} (I + z_{\ast} \sigma_z) \right) =
G(p_H,p_V,z_{\ast}) = G(p_V,p_H,-z_{\ast})$. On the other hand,
$Q_1(\Gamma_{\Lambda}) = q \left( \frac{1}{2} (I + z_{\ast}
\sigma_z) \right)$, which enables us to find
$Q_1(\Gamma_{\Lambda})$ by numerically solving the equation
$G(p_H,p_V,z) = G(p_V,p_H,-z)$ and selecting a solution of a
proper sign. If $\max(p_H,p_V)
> \frac{1}{2}$ and $p_H p_V \neq 0$, then $-1 < z_{\ast} < 1$ and
$Q_1(\Gamma_{\Lambda}) > 0$, which justifies our assumption that
$\max_{\rho \in {\cal D}({\cal H}_2)} q(\rho) > 0$. For the sake
of completeness, we also provide an approximate solution $z'
\approx z_{\ast}$, for which $Q_1(\Gamma_{\Lambda}) = q \left(
\tfrac{1}{2}(I + z_{\ast} \sigma_z) \right) \geq q \left(
\tfrac{1}{2}(I + z' \sigma_z) \right) > 0$, namely,
\begin{eqnarray*}
&& \frac{1-z'}{1+z'} =
\left(\frac{1-p_V}{1-p_H}\right)^{\dfrac{1-p_H}{2p_H - 1}} \left(
\frac{p_H}{p_V} \right)^{\dfrac{p_H}{2p_H-1}} - \frac{(p_H -
p_V)(p_H + p_V - 2 p_H p_V)}{(2p_H - 1)(1-p_V)p_V} \quad \text{if}
\quad p_H \geq p_V, \\
&& \frac{1+z'}{1-z'} =
\left(\frac{1-p_H}{1-p_V}\right)^{\dfrac{1-p_V}{2p_V - 1}} \left(
\frac{p_V}{p_H} \right)^{\dfrac{p_V}{2p_V-1}} - \frac{(p_V -
p_H)(p_H + p_V - 2 p_H p_V)}{(2p_V - 1)(1-p_H)p_H} \quad \text{if}
\quad p_V > p_H.
\end{eqnarray*}

\noindent The heat map of $Q_1(\Gamma_{\Lambda})$ as function of
$p_H$ and $p_V$ is depicted in Fig.~\ref{figure-4}.
\end{example}

\begin{figure}
\includegraphics[width=10cm]{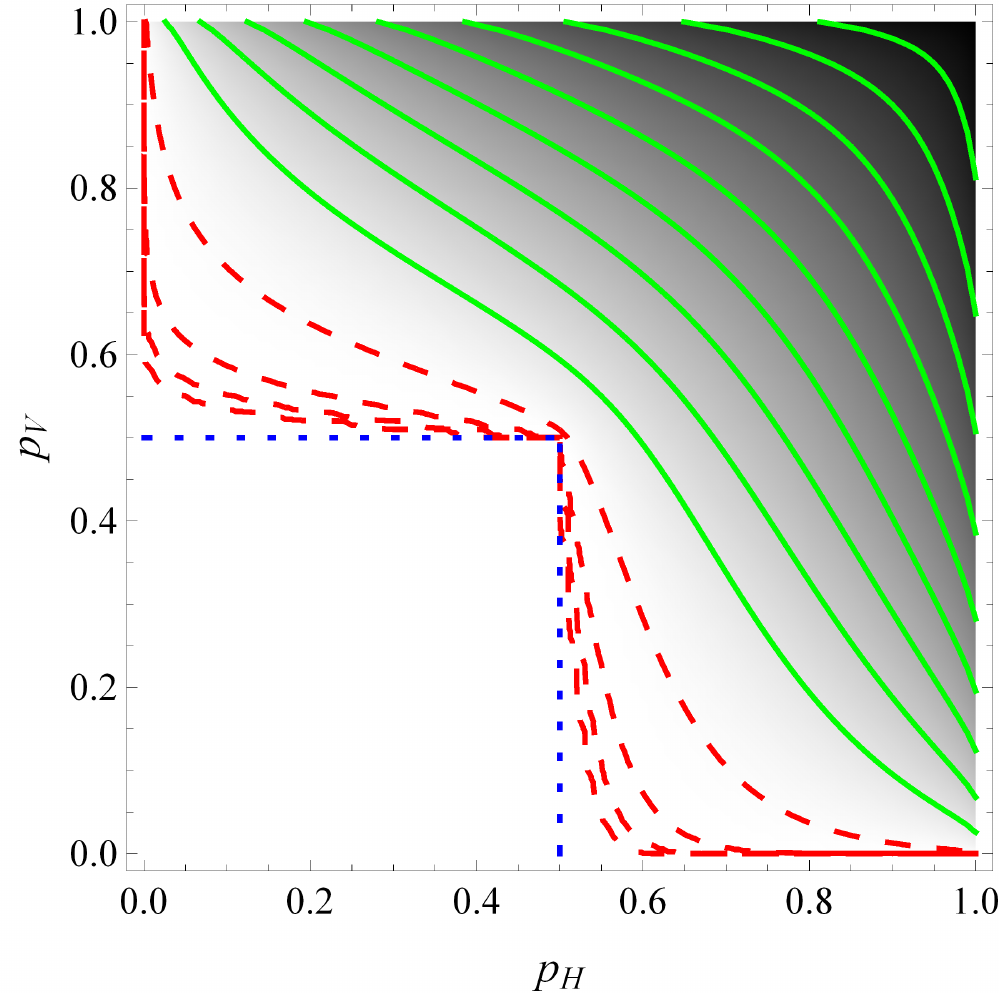}
\caption{\label{figure-4} Heat map of the single-letter quantum
capacity $Q_1(\Gamma_{\Lambda})$ for the quantum operation
$\Lambda$ describing the polarization dependent losses,
Eq.~\eqref{pdl}. Shades of the gray color denote different values
of $Q_1$ in bits; black color represents $1$ and white color
represents $0$. Solid (green) lines correspond to levels $0.9$ to
$0.1$ with the decrement 0.1. Dashed (red) lines correspond to
levels $10^{-2}$, $10^{-4}$, $10^{-6}$, and $10^{-8}$. Dotted
(blue) line denotes the boundary of a region wherein
$Q_1(\Gamma_{\Lambda}) = 0$.}
\end{figure}

In what follows, we study degradability and antidegradability for
a specific class of generalized erasure channels
$\Gamma_{\Lambda}$, where $\Lambda$ has the Kraus rank 1.

\begin{proposition} \label{proposition-degradable}
Suppose the quantum operation $\Lambda \in {\cal O}({\cal H}_d)$
has a single Kraus operator $A$, i.e., $\Lambda[\rho] = A \rho
A^{\dag}$, then $\Gamma_{\Lambda}$ is degradable if and only if
$AA^{\dag} \geq \frac{1}{2}I$ or $I - A^{\dag}A$ is a rank-$1$
operator.
\end{proposition}
\begin{proof}
Using the relation between the Kraus operators for a channel and a
complementary channel~\cite{holevo-2007}, we readily get
\begin{equation} \label{gec-complementary}
\Gamma_{\Lambda}[\rho] = \left(%
\begin{array}{cc}
  A \rho A^{\dag} & {\bf 0} \\
  {\bf 0}^\top & {\rm tr}\left[ \sqrt{I - A^{\dag}A} \, \rho \, \sqrt{I - A^{\dag}A} \right] \\
\end{array}%
\right), \quad \widetilde{\Gamma_{\Lambda}}[\rho] = \left(%
\begin{array}{cc}
  {\rm tr} [ A \rho A^{\dag} ] & {\bf 0}^\top \\
  {\bf 0} &  \sqrt{I - A^{\dag}A} \, \rho \, \sqrt{I - A^{\dag}A} \\
\end{array}%
\right).
\end{equation}

\noindent Up to a proper change of the output basis,
$\widetilde{\Gamma_{\Lambda}}$ coincides with
$\Gamma_{\Lambda'_{\min}}$. Let us consider three distinctive
cases.

(i) Suppose $I - A^{\dag}A$ is a rank-$1$ operator, i.e., $I -
A^{\dag}A = p \ket{\psi}\bra{\psi}$ for some normalized vector
$\ket{\psi} \in {\cal H}_d$ and $p>0$. Then we have
\begin{equation*}
\sqrt{I - A^{\dag}A} \, \rho \, \sqrt{I - A^{\dag}A} = p
\bra{\psi} \rho \ket{\psi} \ket{\psi}\bra{\psi} = {\rm tr}\left[
\sqrt{I - A^{\dag}A} \, \rho \, \sqrt{I - A^{\dag}A} \right]
\ket{\psi}\bra{\psi}
\end{equation*}

\noindent and $\widetilde{\Gamma_{\Lambda}}$ is degradable because
$\widetilde{\Gamma_{\Lambda}} = \Xi \circ \Gamma_{\Lambda}$, where
the quantum channel $\Xi$ reads
$ \Xi \left[ \left(%
\begin{array}{cc}
  \rho & \vdots \\
  \cdots & c \\
\end{array}%
\right) \right] = \left(%
\begin{array}{cc}
  {\rm tr}[\rho] & {\bf 0}^{\top} \\
  {\bf 0} & c \ket{\psi}\bra{\psi} \\
\end{array}%
\right)$.

(ii) Suppose $I - A^{\dag}A$ is not a rank-$1$ operator. Then
$\sqrt{I - A^{\dag}A} \, \rho \, \sqrt{I - A^{\dag}A}$ generally
has the rank greater than or equal to 2, so this operator cannot
be obtained from a linear map acting on ${\rm tr}\left[ \sqrt{I -
A^{\dag}A} \, \rho \, \sqrt{I - A^{\dag}A} \right]$. Therefore,
the operator $\sqrt{I - A^{\dag}A} \, \rho \, \sqrt{I -
A^{\dag}A}$ should be obtained from a linear map acting on $A \rho
A^{\dag}$. Suppose $\det A \neq 0$, then there exists a unique
linear map that for all $\rho \in {\cal B}({\cal H}_d)$ maps the
operator $A\rho A^{\dag}$ to the operator $\sqrt{I - A^{\dag}A} \,
\rho \, \sqrt{I - A^{\dag}A}$. Therefore, if $\det A \neq 0$, then
there exists a unique linear map $\Xi$ such that
$\widetilde{\Gamma_{\Lambda}} = \Xi \circ \Gamma_{\Lambda}$. It
reads
\begin{equation*}
\Xi \left[ \left(%
\begin{array}{cc}
  \rho & \vdots \\
  \cdots & c \\
\end{array}%
\right) \right] = \left(%
\begin{array}{cc}
  c + {\rm tr} \{ \rho [2 I - (AA^{\dag})^{-1}] \} & {\bf 0}^{\top} \\
  {\bf 0} & \sqrt{I - A^{\dag}A} \, A^{-1} \rho (A^{\dag})^{-1} \sqrt{I - A^{\dag}A} \\
\end{array}%
\right).
\end{equation*}

\noindent It is not hard to see that $\Xi$ is trace preserving;
however, $\Xi$ is completely positive if and only if $2 I -
(AA^{\dag})^{-1} \geq 0$, which is equivalent to $AA^{\dag} \geq
\frac{1}{2}I$. On the other hand, if an operator $A$ satisfies
$AA^{\dag} \geq \frac{1}{2}I$ then $\det A \neq 0$ automatically.

(iii) Suppose $I - A^{\dag}A$ is not a rank-$1$ operator and $\det
A = 0$. If $A=0$, then $\Gamma_{\Lambda}$ is obviously not
degradable, so in what follows we additionally assume ${\rm supp}
A \neq \emptyset$. Let $\ket{f} \in {\rm ker} A$ and $\ket{g} \in
{\rm supp} A$, then $\Gamma_{\Lambda}[\ket{f}\bra{g}] = 0$ but
$\widetilde{\Gamma_{\Lambda}}[\ket{f}\bra{g}] = 0 \oplus
\ket{f}\bra{g} \sqrt{I - A^{\dag}A}$, so
$\widetilde{\Gamma_{\Lambda}}[\ket{f}\bra{g}] = 0$ if and only if
$A^{\dag}A\ket{g} = \ket{g}$. Therefore, the degradability of
$\Gamma_{\Lambda}$ implies $A^{\dag}A \ket{g} = \ket{g}$ for all
$\ket{g} \in {\rm supp}A$, i.e., $A^{\dag}A$ is to be a projector.
Since $\det A = 0$, the rank of the projector $A^{\dag}A$ is
bounded from above by $d-1$. If ${\rm rank} A^{\dag}A \leq d-2$,
then there exist two orthonormal vectors $\ket{f_1},\ket{f_2} \in
{\rm ker}A$ and $\Gamma_{\Lambda}[\ket{f_1}\bra{f_2}] = 0$ whereas
$\widetilde{\Gamma_{\Lambda}} [\ket{f_1}\bra{f_2}] = 0 \oplus
\ket{f_1}\bra{f_2} \neq 0$. Therefore, the degradability of
$\Gamma_{\Lambda}$ implies $A^{\dag}A$ is a projector of rank
$d-1$. This contradicts the assumption that $I - A^{\dag}A$ is not
a rank-$1$ operator.
\end{proof}

\begin{proposition} \label{proposition-antidegradable}
Suppose the quantum operation $\Lambda \in {\cal O}({\cal H}_d)$
has a single Kraus operator $A$, i.e., $\Lambda[\rho] = A \rho
A^{\dag}$, then $\Gamma_{\Lambda}$ is antidegradable if and only
if $A^{\dag}A \leq \frac{1}{2}I$ or $A^{\dag}A$ is a rank-1
operator.
\end{proposition}
\begin{proof}
$\Gamma_{\Lambda}$ and $\widetilde{\Gamma_{\Lambda}}$ are given by
Eq.~\eqref{gec-complementary}. Antidegradability of
$\Gamma_{\Lambda}$ is equivalent to degradability of
$\widetilde{\Gamma_{\Lambda}}$. The change $B = \sqrt{I -
A^{\dag}A}$ leads to the relation
\begin{equation} \label{gec-complementary-B}
\widetilde{\Gamma_{\Lambda}}[\rho] = \left(%
\begin{array}{cc}
  {\rm tr} \left[ \sqrt{I - B^{\dag}B} \, \rho \, \sqrt{I - B^{\dag}B} \right] & {\bf 0}^\top \\
  {\bf 0} &  B \, \rho \, B^{\dag} \\
\end{array}%
\right).
\end{equation}

\noindent Therefore, $\widetilde{\Gamma_{\Lambda}}$ is unitarily
equivalent to the generalized erasure channel $\Gamma_{\Upsilon}$,
where $\Upsilon[\rho] = B \rho B^{\dag}$. By
Proposition~\ref{proposition-degradable} $\Gamma_{\Upsilon}$ is
degradable if and only if $BB^{\dag} \geq \frac{1}{2}I$ or $I -
B^{\dag}B$ is a rank-$1$ operator. Substituting $B = \sqrt{I -
A^{\dag}A}$ into these relations, we get that
$\widetilde{\Gamma_{\Lambda}}$ is degradable if and only if
$A^{\dag}A \leq \frac{1}{2}I$ or $A^{\dag}A$ is a rank-1 operator.
\end{proof}

\begin{example} \label{example-degradability}
Let $\Lambda \in {\cal O}({\cal H}_2)$ be a quantum operation
describing polarization dependent losses, Eq.~\eqref{pdl}. As the
Kraus rank of $\Lambda$ equals 1, we apply
Propositions~\ref{proposition-degradable}
and~\ref{proposition-antidegradable} and obtain the following
results:

(i) $\Gamma_{\Lambda}$ is degradable if and only if $\min(p_H,p_V)
\geq \frac{1}{2}$ or $p_H = 1$ or $p_V = 1$;

(ii) $\Gamma_{\Lambda}$ is antidegradable if and only if
$\max(p_H,p_V) \leq \frac{1}{2}$ or $p_H=0$ or $p_V=0$.

Therefore, $Q(\Gamma_{\Lambda}) = Q_1(\Gamma_{\Lambda})$ if
$\min(p_H,p_V) \geq \frac{1}{2}$ or $p_H = 1$ or $p_V = 1$ and,
moreover, we exactly know $Q(\Gamma_{\Lambda})$ thanks to the
result of Example~\ref{example-Q1}. Additionally, we know that
$Q(\Gamma_{\Lambda}) = 0$ if $\max(p_H,p_V) \leq \frac{1}{2}$ or
$p_H=0$ or $p_V=0$.
\end{example}

For $\Lambda$ in Eq.~\eqref{pdl},
Example~\ref{example-degradability} leaves $Q(\Gamma_{\Lambda})$
uncertain in two regions of parameters, where either $ \frac{1}{2}
< p_H < 1$ and $0 < p_V < \frac{1}{2}$ or $0 < p_H < \frac{1}{2}$
and $\frac{1}{2} < p_V < 1$. The following result shows that in
half of this region the superadditivity of coherent information
takes place.

\begin{figure}
\includegraphics[width=10cm]{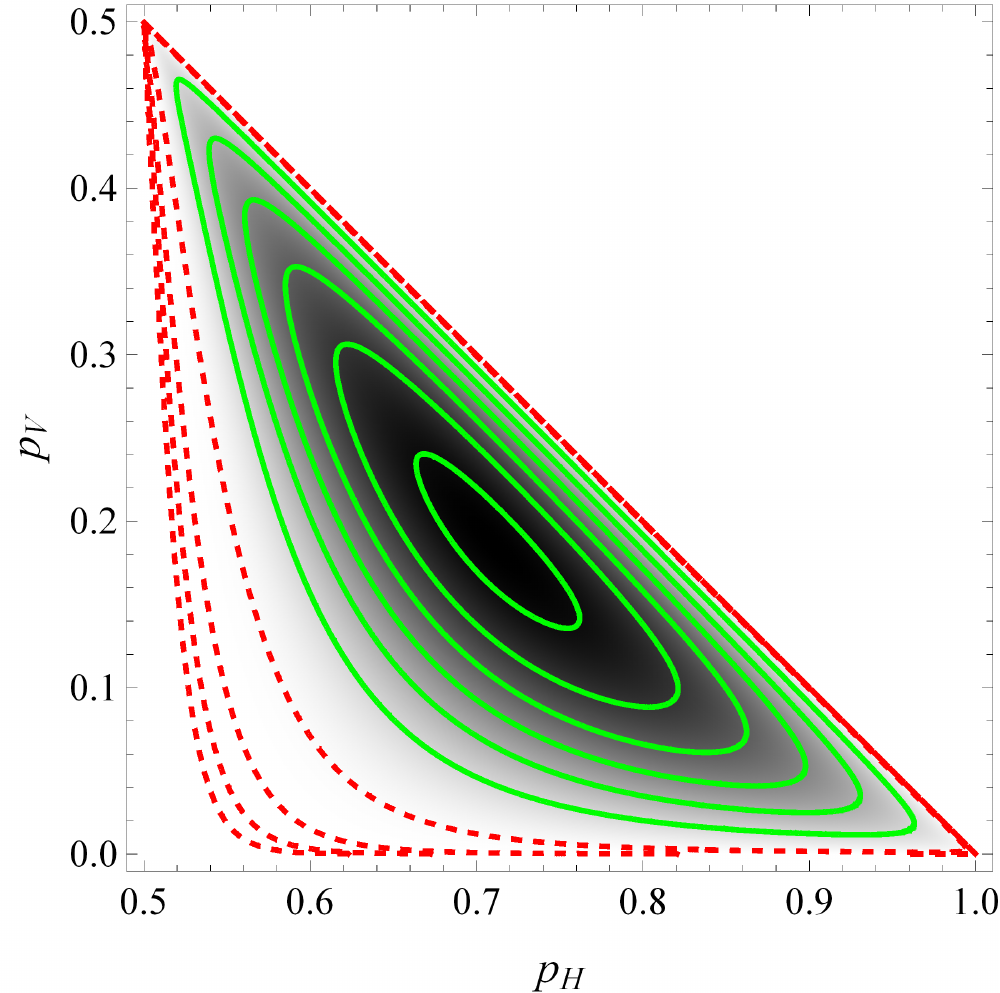}
\caption{\label{figure-5} Heat map of the lower bound for
$\frac{1}{2} Q_1(\Gamma_{\Lambda}^{\otimes 2}) -
Q_1(\Gamma_{\Lambda})$, where $\Lambda$ describes the polarization
dependent losses, Eq.~\eqref{pdl}. Shades of the gray color denote
different values of the lower bound in bits; black color
represents $7 \cdot 10^{-3}$ and white color represents $0$. Solid
(green) lines correspond to levels $6 \cdot 10^{-3}$ to $10^{-3}$
with the decrement $10^{-3}$. Dashed (red) lines correspond to
levels $10^{-4}$, $10^{-6}$, $10^{-8}$, and $10^{-10}$.}
\end{figure}

\begin{proposition} \label{proposition-q2-q1}
Let $\Lambda \in {\cal O}({\cal H}_2)$ be a quantum operation
defined by Eq.~\eqref{pdl} and describing polarization dependent
losses with parameters $p_H$ and $p_V$. The strict inequality
$\frac{1}{2} Q_1(\Gamma_{\Lambda}^{\otimes 2}) >
Q_1(\Gamma_{\Lambda})$ holds if either $\frac{1}{2} < p_H < 1$ and
$0 < p_V < 1 - p_H$ or $\frac{1}{2} < p_V < 1$ and $ 0 < p_H < 1 -
p_V$.
\end{proposition}
\begin{proof}
Let $\rho_1 =  \rho_{HH} \ket{H}\bra{H} + \rho_{VV} \ket{V}\bra{V}
\in {\cal D}({\cal H}_2)$ be a density operator for which
$Q_1(\Gamma_{\Lambda}) = S(\Gamma_{\Lambda}[\rho_1]) -
S(\widetilde{\Gamma_{\Lambda}}[\rho_1])$, i.e., $\rho_{HH} =
\frac{1+z_{\ast}}{2}$ and $\rho_{VV} = \frac{1-z_{\ast}}{2}$,
where $z_{\ast}$ is a solution of the equation $G(p_H,p_V,z) =
G(p_V,p_H,-z)$ such that ${\rm sgn} (z_{\ast}) = {\rm sgn} (p_V -
p_H)$, see Example~\ref{example-Q1}. Consider the following
operator $\rho_2 \in {\cal D}({\cal H}_4)$:
\begin{equation}
\rho_2 = \rho_{HH}^2 \ket{HH}\bra{HH} + 2 \rho_{HH} \rho_{VV}
\ket{\varphi_{-}}\bra{\varphi_{-}} + \rho_{VV}^2 \ket{VV}\bra{VV},
\quad \ket{\varphi_{-}} = \frac{1}{\sqrt{2}} (\ket{HV} -
\ket{VH}),
\end{equation}

\noindent where $\ket{\varphi_{-}}\bra{\varphi_{-}}$ is an
entangled pure state. In the basis
$\{\ket{HH},\ket{HV},\ket{VH},\ket{VV}\}$ the diagonal of $\rho_2$
is exactly the diagonal of the diagonal matrix $\rho_1^{\otimes
2}$. Since both $\Lambda$ and $\Lambda_{\min}'$ have a single
diagonal Kraus operator, the application of maps $\Lambda^{\otimes
2}$, $\Lambda \otimes \Lambda_{\min}'$, $\Lambda_{\min}' \otimes
\Lambda$, and $(\Lambda_{\min}')^{\otimes 2}$ preserves the
positions of non-zero elements in the matrices $\rho_2$ and
$\rho_1^{\otimes 2}$. Recalling the definition of the
trash-and-prepare channel, ${\rm Tr}[\rho] = {\rm tr}[\rho]
\ket{e}\bra{e}$, we have
\begin{eqnarray*}
&& \Lambda \otimes ({\rm Tr} \circ \Lambda'_{\min}) [\rho_2] =
\Lambda \otimes ({\rm Tr} \circ \Lambda'_{\min}) [\rho_1^{\otimes
2}], \\
&& ({\rm Tr} \circ \Lambda'_{\min}) \otimes \Lambda [\rho_2] =
({\rm Tr} \circ \Lambda'_{\min}) \otimes \Lambda [\rho_1^{\otimes
2}], \\
&& ({\rm Tr} \circ \Lambda'_{\min})^{\otimes 2} [\rho_2] = ({\rm
Tr} \circ \Lambda'_{\min})^{\otimes 2} [\rho_1^{\otimes 2}].
\end{eqnarray*}

\noindent It follows from Eq.~\eqref{Gamma-Lambda-tensor-power-2}
that the only difference between $\Gamma_{\Lambda}^{\otimes 2}
[\rho_2]$ and $\Gamma_{\Lambda}^{\otimes 2} [\rho_1^{\otimes 2}]$
is in the blocks ${\Lambda}^{\otimes 2} [\rho_2]$ and
${\Lambda}^{\otimes 2} [\rho_1^{\otimes 2}]$. Moreover, within
these blocks the difference is present only in $2 \times 2$
submatrices, namely, the submatrix $2 p_H p_V \rho_{HH} \rho_{VV}
\ket{\varphi_{-}}\bra{\varphi_{-}}$ and the submatrix $p_H p_V
\rho_{HH} \rho_{VV}(\ket{HV}\bra{HV} + \ket{VH}\bra{VH})$ for
${\Lambda}^{\otimes 2} [\rho_2]$ and ${\Lambda}^{\otimes 2}
[\rho_1^{\otimes 2}]$, respectively. Since ${\rm Spec}(2
\ket{\varphi_{-}}\bra{\varphi_{-}}) = \{2,0\}$ and ${\rm
Spec}(\ket{HV}\bra{HV} + \ket{VH}\bra{VH}) = \{1,1\}$, we
explicitly relate the entropies as follows:
\begin{equation}
S \big( \Gamma_{\Lambda}^{\otimes 2}[\rho_2] \big) = S \big(
\Gamma_{\Lambda}^{\otimes 2}[\rho_1^{\otimes 2}] \big) - 2 p_H p_V
\rho_{HH} \rho_{VV} \log 2. \label{S-decrease-direct}
\end{equation}

\noindent Analogous consideration for the complementary channel
yields
\begin{equation}
S \big( \widetilde{\Gamma_{\Lambda}}^{\otimes 2}[\rho_2] \big) = S
\big( \widetilde{\Gamma_{\Lambda}}^{\otimes 2}[\rho_1^{\otimes 2}]
\big) - 2 (1-p_H)(1-p_V) \rho_{HH} \rho_{VV} \log 2.
\label{S-decrease-complementary}
\end{equation}

\noindent Since $S \big( \Gamma_{\Lambda}^{\otimes
2}[\rho_1^{\otimes 2}] \big) = 2 S (\Gamma_{\Lambda}[\rho_1])$ and
$S \big( \widetilde{\Gamma_{\Lambda}}^{\otimes 2}[\rho_1^{\otimes
2}] \big) = 2 S(\widetilde{\Gamma_{\Lambda}}[\rho_1])$, we readily
obtain the following lower bound for the two-letter quantum
capacity:
\begin{eqnarray*}
\frac{1}{2}Q_1(\Gamma_{\Lambda}^{\otimes 2}) & \geq & \frac{1}{2}
\left[ S(\Gamma_{\Lambda}^{\otimes 2}[\rho_2]) -
S(\widetilde{\Gamma_{\Lambda}}^{\otimes 2}[\rho_2]) \right] \nonumber\\
& = & S(\Gamma_{\Lambda}[\rho_1]) -
S(\widetilde{\Gamma_{\Lambda}}[\rho_1]) + (1-p_H-p_V)
\rho_{HH} \rho_{VV} \log 2 \nonumber\\
& = & Q_1(\Gamma_{\Lambda}) + (1-p_H-p_V) \rho_{HH} \rho_{VV} \log
2.
\end{eqnarray*}

\noindent If $p_H$ and $p_V$ satisfy the requirements in the
statement of Proposition~\ref{proposition-q2-q1}, then $1-p_H-p_V
> 0$ and $-1 < z_{\ast} < 1$, which implies $(1-p_H-p_V) \rho_{HH}
\rho_{VV} > 0$.
\end{proof}

In Fig.~\ref{figure-5} we depict the derived lower bound
$(1-p_H-p_V) \rho_{HH} \rho_{VV}$ bits for the difference
$\frac{1}{2} Q_1(\Gamma_{\Lambda}^{\otimes 2}) -
Q_1(\Gamma_{\Lambda})$ in the region of parameters $\frac{1}{2} <
p_H < 1$ and $0 < p_V < \frac{1}{2}$. Numerics show that the
actual difference $\frac{1}{2} Q_1(\Gamma_{\Lambda}^{\otimes 2}) -
Q_1(\Gamma_{\Lambda})$ has a similar shape within the specified
region and vanishes (up to a machine precision) if $p_H + p_V \geq
1$. The maximum achievable difference $\frac{1}{2}
Q_1(\Gamma_{\Lambda}^{\otimes 2}) - Q_1(\Gamma_{\Lambda})$
approximately equals $7.197 \cdot 10^{-3}$ and is achieved in the
vicinity of parameters $p_H = 0.7$ and $p_V = 0.19$ (or vice
versa).

Physical meaning of Eqs.~\eqref{S-decrease-direct}
and~\eqref{S-decrease-complementary} is that the use of $\rho_2$
instead of $\rho_1^{\otimes 2}$ in the two-letter scenario
diminishes both the entropy of the channel output and the entropy
of the complementary channel output. However, the decrement in
Eq.~\eqref{S-decrease-direct} is less that the decrement in
Eq.~\eqref{S-decrease-complementary}, i.e., less information is
dissolved into environment and more information reaches the
receiver as compared to the single-letter case. Despite the fact
that the losses are asymmetric, i.e., $p_H \neq p_V$, the
contribution $\ket{\varphi_-}\bra{\varphi_-}$ in $\rho_2$
preserves its form in the output states $\Gamma_{\Lambda}^{\otimes
2}[\rho_2]$ and $\widetilde{\Gamma_{\Lambda}}^{\otimes 2}[\rho_2]$
because the same product $p_H p_V$ characterizes the transmission
of both $HV$ and $VH$ pairs of photons.

\section{Conclusions} \label{section-conclusions}

We reviewed physical properties of trace decreasing quantum
operations and clarified a distinction between biased and unbiased
quantum operations. We emphasized the importance of biased quantum
operations and motivated the introduction of the generalized
erasure channel. We identified information capacities of a trace
decreasing quantum operation with the corresponding capacities of
the generalized erasure channel.

As to general mathematical results, we proved some simple yet
fruitful characterizations for extensions of a quantum operation
to a channel (Proposition~\ref{proposition-extension}) and the
normalized image of a trace decreasing operation
(Proposition~\ref{proposition-images}). The channel
$\Phi_{\Lambda}$ found in Proposition~\ref{proposition-images} was
subsequently used in finding lower and upper bounds for the
single-letter classical and quantum capacities of the generalized
erasure channel
(Propositions~\ref{proposition-C1-lower-upper-bounds} and
\ref{proposition-Q1-lower-upper-bound}). Bounds on the regularized
classical and quantum capacities of the generalized erasure
channel were expressed through the minimal and maximal detection
probabilities (Propositions~\ref{proposition-C-upper-bound},
\ref{proposition-C-lower-bound}, and
\ref{proposition-Q-upper-bound}). We showed that the biasedness of
a quantum operation automatically guarantees nonzero classical
capacity of the generalized erasure channel
(Proposition~\ref{proposition-C-lower-bound}). For quantum
operations with Kraus rank 1 we fully characterized necessary and
sufficient conditions for degradability and antidegradability of
the corresponding generalized erasure channel
(Propositions~\ref{proposition-degradable} and
\ref{proposition-antidegradable}).

As a prominent physical example of a biased quantum operation we
considered polarization dependent losses. In addition to the
calculation of the single-letter quantum capacity for that
physical situation in Example~\ref{example-Q1}, we managed to
provide an analytical proof for the superadditivity of coherent
information, i.e., a strict separation between the single-letter
quantum capacity and the two-letter quantum capacity
(Proposition~\ref{proposition-q2-q1}). Importantly, the observed
difference $\frac{1}{2} Q_1(\Gamma_{\Lambda}^{\otimes 2}) -
Q_1(\Gamma_{\Lambda})$ was shown to achieve $7.197 \cdot 10^{-3}$
bits per qubit sent, which is the maximum reported value for
superadditivity of coherent information among qubit-input
channels. These results show that the polarization dependent
losses may serve as a testbed for exploring other interesting
effects, for instance, checking the superadditivity of private
information.

\begin{acknowledgements}
The author thanks Vikesh Siddhu for useful comments and the
anonymous referee for valuable comments to improve the quality of
the manuscript. The study was supported by the Russian Science
Foundation, project no. 19-11-00086.
\end{acknowledgements}

\end{document}